\documentclass{iopjournal}

\usepackage{amsmath}
\usepackage{amssymb}
\usepackage{cite}
\usepackage{url}
\fancyhfoffset[L]{0mm}
\newcommand{\loosepar}{\tolerance=9999\emergencystretch=3em\relax}

\usepackage{amsthm}
\usepackage{longtable}
\usepackage{tabularx}
\usepackage{booktabs}
\usepackage{tikz}
\usetikzlibrary{positioning, fit, calc, backgrounds, arrows.meta,
 decorations.pathmorphing, decorations.markings,
 decorations.pathreplacing}

\newcommand{\Cl}{\mathrm{Cl}}
\newcommand{\Spin}{\mathrm{Spin}}
\newcommand{\SOgroup}{\mathrm{SO}}
\newcommand{\SLgroup}{\mathrm{SL}}
\newcommand{\Vflav}{V_{\mathrm{flav}}}
\newcommand{\evMV}{\mathcal{E}}
\newcommand{\grade}[2]{\langle #1\rangle_{#2}}
\newcommand{\reverse}[1]{\widetilde{#1}}
\newcommand{\Hlhc}{H_{\mathrm{LHC}}}
\newcommand{\MET}{E_T^{\mathrm{miss}}}
\newcommand{\netRef}{\textsc{Ref}}
\newcommand{\netGa}{\textsc{Ga}}

\newtheorem{theorem}{Theorem}

\begin{document}

\title{Geometric algebra as the input language of collider foundation models\textsuperscript{*}}

\author{E.~Abasov$^{1}$, L.~Dudko$^{1,\dagger}$, F.~Grigoryev$^{1}$, P.~Volkov$^{1}$
and A.~Zaborenko$^{1}$}

\affil{$^{1}$Skobeltsyn Institute of Nuclear Physics, Lomonosov Moscow State
University (SINP MSU), 1(2) Leninskie gory, GSP-1, Moscow 119991, Russian Federation}

\affil{$^{\dagger}$Author to whom any correspondence should be addressed.}

\email{dudko@sinp.msu.ru}

{\renewcommand{\thefootnote}{\fnsymbol{footnote}}%
\footnotetext[1]{This study was conducted within the scientific programme of the National Center for Physics and Mathematics, section \#5 ``Particle Physics and Cosmology''. Stage 2026--2027.}}

\begin{abstract}
We describe a hard hadron-collider event as a single geometric object. Each
reconstructed final-state object is a multivector of $\Cl(1,3)\otimes\Vflav$
that combines its four-momentum with its object-type label. The event
multivector $\evMV$ is the direct sum of these per-object multivectors.
Products in the algebra generate the higher grades from the measured
four-momenta. The grades organise nearly every observable in current use, from invariants at
grade zero to decay-plane bivectors, oriented three-volumes and the CP-odd
pseudoscalar at grades two to four. The high-level invariants and the low-level four-momentum recipe of
earlier analyses, and the inputs of equivariant networks, are recovered as
projections onto specific grades. We give a per-grade dictionary of $32$
classical observables and an inventory of the symmetries acting on $\evMV$.
A theorem extending the Cayley--Menger identity to Minkowski space shows
that the higher grades give no Lorentz invariant beyond $\{p_i\!\cdot\!p_j,\,m_i^2\}$ except the CP-odd sign of the pseudoscalar. Grades two and three nevertheless
describe the decay planes and decay volumes of the event compactly, as
covariant objects. The representation is intended as a uniform input layer for
foundation models of collider physics. As a simple example of its use, we apply it to the resonance-topology
separation of $pp\!\to\!tWb$ with a small Lorentz-equivariant multivector
transformer of L-GATr type. In an ablation over eight input representations of this example, the
classifier improves when a fixed reference bivector encoding the beam plane,
which the momenta do not contain, is added to the input.
\end{abstract}


\section{Introduction}
\label{sec:intro}

The performance and interpretability of a neural network trained on collider
data depend on its input representation at least as much as on its
architecture. We propose an input representation that describes a hard
hadron-collider event as a single geometric object. Each
reconstructed object (a lepton, a jet, a photon or the missing transverse
momentum) is a multivector. Its four-momentum is a grade-one element of the
geometric algebra $\Cl(1,3)$ of
spacetime~\cite{Doran:2003,Hestenes:1984,Lounesto:2001}, and its type label is a basis
vector of an object-type space $\Vflav$. The event multivector $\evMV$ is the
direct sum of these per-object multivectors, with one
$\Cl(1,3)\otimes\Vflav$ slot per object (Eq.~\eqref{eq:evMV}). All slots
share the same algebra, so their momenta can be multiplied in it. The
products give the invariants, decay planes, oriented three-volumes and
pseudoscalars of the event as grades zero to four of one graded multivector
(Sec.~\ref{sec:evmv-def}).

The usual input is instead a list of four-momenta
with discrete labels for charge, $b$-tag, lepton flavour and missing
transverse momentum. The high-level invariant recipes~\cite{Boos:2008sdz}, the
low-level four-momentum recipes~\cite{Dudko:2020qas} and the equivariant
networks of Refs.~\cite{Bogatskiy:2020tje, Bogatskiy:2022czk,
brehmer2023geometric, de2023euclidean, Spinner:2024hjm, Brehmer:2024yqw}
each use only a part of this object. The first two choose \emph{which
scalar features to feed the network}, and the third makes the network
respect Lorentz symmetry without fixing the features. In all three the
discrete reconstructed information stays outside the kinematic vector space.

Geometric algebra is the appropriate framework for this view, because its grade decomposition already
organises nearly all observables used in collider analyses. The Lorentz-invariant
scalar inputs of Refs.~\cite{Boos:2008sdz, Dudko:2020qas} form the
grade-zero sector of $\evMV$. Lab-frame variables are components of the grade-one sector with respect to the
reference vectors that define the frame, and helicity-frame angles are
grade-zero projections. The raw covariant coefficients of the higher grades, which
L-GATr~\cite{Brehmer:2024yqw} carries in its hidden channels, are the input
suited to Lorentz-equivariant networks. The construction is a common notational
framework and adds no new physical content.

As the input of a foundation model of collider events, the single algebraic
framework of $\Cl(1,3)$ has three advantages over a concatenation of features
from different recipes. First, $\Spin^+(1,3)$ acts within the
algebra, and every equivariant bilinear map on multivector inputs is a grade
projection of the geometric product or of its pseudoscalar-dressed
counterpart $\grade{x I y}{k}$ (Sec.~\ref{sec:ops}). Here $\grade{X}{k}$ is
the grade-$k$ part of a multivector $X$. The algebra itself thus provides
the inductive bias of Lorentz-covariant networks. Second, the higher grades add no Lorentz-invariant scalar beyond
the inner products and masses $\{p_i\!\cdot\!p_j,\,m_i^2\}$ except the
one-bit CP-odd sign of the pseudoscalar (Theorem~\ref{lemma:cm-collapse}).
The theorem extends the Euclidean Cayley--Menger (Gram-determinant)
identity~\cite{Cayley:1841, Menger:1928} to Lorentzian signature (cf.\ the
simplex-realisability analysis of Ref.~\cite{Tate:2011rm}), and
Sec.~\ref{sec:grades} relates the CP-odd sign to earlier CP analyses~\cite{Atwood:2000tu}. The
theorem restricts the invariant scalars only. Grades two and three describe
the decay planes and decay volumes of the event compactly, as covariant
objects whose components do not reduce to inner products. Third, the grade decomposition
(Eq.~\eqref{eq:grade-decomp}) gives the same multivector type $\evMV$ and
the same per-grade interface to every process and every downstream task,
from classification and regression to CP-asymmetry tests
(Sec.~\ref{sec:outlook}).

The representation, Theorem~\ref{lemma:cm-collapse} and the grade-resolved
fine-tuning interface do not depend on a particular network.
Section~\ref{sec:tWb} gives a simple example of the use of the representation,
the separation of the single- and double-resonant topologies of $pp\to
tWb$~\cite{Boos:2023kpp,Boos:2020rqy,Baskakov:2019rbq} with a small
single-task network. The example illustrates the input layer and is not an
implementation of a foundation model, which is the subject of the programme
outlined in Sec.~\ref{sec:outlook}. In this network a \emph{token} is one
input slot and holds a
stack of channels, each a full $16$-component multivector of $\Cl(1,3)$. The
$1{+}4{+}6{+}4{+}1$ grade cells of a channel are \emph{storage}, filled only
where the input supplies content. The particle tokens are the slots of
Eq.~\eqref{eq:evMV}, and event-level and reference tokens extend the
sequence in the same format. Which cells are populated is the input choice
under test (App.~\ref{app:tokens} gives the layouts cell by cell).

Each reconstructed object $i$ enters as one multivector token
$T_i = \grade{T_i}{0}\oplus\grade{T_i}{1}$, as in
L-GATr~\cite{Brehmer:2024yqw}. Two event-level pairing tokens supply the grade-two and
grade-three objects of the candidate top decays directly at the input layer.
In this small example the pairing tokens give no gain, since the network forms
the few higher-grade objects it needs in its own layers
(Sec.~\ref{sec:ablation}). A fixed reference bivector encoding the beam plane,
which cannot be derived from the measured momenta, improves the classifier.
The class difference of the example is carried mainly by one resonant
three-particle combination, the leptonic top decay. Its geometry is given by
the grade-two decay planes of its two-body subsystems and by its grade-three
oriented decay volume. A foundation model has to cover many final states with
several different resonances. There the relevant combinations are harder to
find, and higher grades computed at the input could improve the efficiency of
the network.

Section~\ref{sec:primer} fixes conventions, defines $\evMV$ and relates the
construction to symmetry-aware architectures.
Section~\ref{sec:grades} gives the grade-by-grade dictionary of the
kinematic factor, Sec.~\ref{sec:ops} the operations on $\evMV$ with their
physical meaning, and Sec.~\ref{sec:symm} the symmetry inventory.
Section~\ref{sec:flav} describes the object-type factor $\Vflav$ and the
treatment of missing transverse momentum. Section~\ref{sec:tWb} presents the
$pp\to tWb$ demonstration, and Sec.~\ref{sec:outlook} outlines the use of
$\evMV$ as the input of collider foundation models~\cite{Mikuni:2024qsr}.
Section~\ref{sec:limits} discusses limitations and future directions, and
Sec.~\ref{sec:summary} summarises.

\section{Geometric algebra in $\Cl(1,3)$}
\label{sec:primer}

The paper uses a small part of the geometric algebra $\Cl(1,3)$ of Minkowski
spacetime~\cite{Doran:2003,Hestenes:1984,Lounesto:2001}, summarised below.
The full notation card is in App.~\ref{app:notation}, and the identities
used in the body are derived in App.~\ref{app:derivations}.

\paragraph{Conventions.}
The metric signature is $(+,-,-,-)$, and the orthonormal grade-one basis
$\{\gamma_0,\gamma_1,\gamma_2,\gamma_3\}$ satisfies
\begin{equation}
 \gamma_\mu \gamma_\nu + \gamma_\nu \gamma_\mu
 \;=\; 2\, \eta_{\mu\nu},
 \qquad \eta = \mathrm{diag}(+,-,-,-).
 \label{eq:cliff-anticomm}
\end{equation}
The geometric product of two grade-one elements splits into a symmetric
inner part reproducing the Minkowski metric,
$p\!\cdot\!q = \tfrac12(pq+qp)$, and an antisymmetric outer part,
$p\wedge q = \tfrac12(pq-qp)$, so that $pq = p\!\cdot\!q + p\wedge q$
(Eq.~\eqref{eq:gp-split}, component form in Eq.~\eqref{eq:gp-components}).
The proper orthochronous Lorentz group acts on multivectors through the
rotor sandwich
\begin{equation}
 X \;\mapsto\; R\, X\, \reverse{R},
 \qquad
 R \;=\; \exp(B/2),
 \quad B \in \grade{\Cl(1,3)}{2},
 \quad R\,\reverse{R} = 1,
 \label{eq:rotor-sandwich}
\end{equation}
with $R$ in the spin cover $\Spin^+(1,3) \simeq \SLgroup(2,\mathbb C)$. We
also use the pseudoscalar $I = \gamma_0\gamma_1\gamma_2\gamma_3$
($I^2 = -1$), the reverse $\reverse{X}$, which reverses the order of the
vector factors, the grade projection $\grade{X}{k}$ and the Hodge-type dual
$X^\star \equiv X\, I^{-1}$ with its inverse $X^{-\star} \equiv X\, I$. The
notation card defines them and gives the bivector sign rules.

\paragraph{The algebra $\Cl(1,3)$.}
The algebra splits into homogeneous components by grade,
\begin{equation}
 \Cl(1,3) \;=\; \bigoplus_{k=0}^{4}\, \grade{\Cl(1,3)}{k},
 \qquad
 \dim \grade{\Cl(1,3)}{k} \;=\; \binom{4}{k},
 \qquad
 1 + 4 + 6 + 4 + 1 \;=\; 16.
 \label{eq:grade-decomp}
\end{equation}
Each grade has a direct physical meaning (Fig.~\ref{fig:grade-ladder}).
Grade zero is a Lorentz scalar, e.g.\ an inner product
$p_i\!\cdot\!p_j$, grade one is a four-momentum, and grade two is an
oriented decay plane $p_i\wedge p_j$. Grade three is an oriented
three-volume $p_i\wedge p_j\wedge p_k$, equivalently a covariant four-vector
via the dual $\star$. Grade four is the pseudoscalar coefficient
$\epsilon_{\mu\nu\rho\sigma}\, p_1^\mu p_2^\nu p_3^\rho p_4^\sigma$, with
$\epsilon_{\mu\nu\rho\sigma}$ the totally antisymmetric Levi-Civita symbol.
It changes sign under a spatial inversion and is C-even, hence CP-odd
(Sec.~\ref{sec:grades}). It also changes sign under an odd permutation of
the four momenta. For an $r$-blade $P_r$ built from $r$ four-momenta, the
reversion contraction $\grade{P_r\reverse{P_r}}{0}$ equals the $r\times r$
Gram determinant of their inner products (Theorem~\ref{lemma:cm-collapse}).

The opposite signature $(-,+,+,+)$ gives $\Cl(3,1) \simeq M_4(\mathbb R)$,
which is not isomorphic to $\Cl(1,3) \simeq M_2(\mathbb H)$ as a real
associative algebra~\cite[Tab.~16.4]{Lounesto:2001}. The two share the even
subalgebra $\Cl^+(1,3) \simeq \Cl^+(3,1)$ and the spin cover, so they
describe the same observables, and we use $\Cl(1,3)$ as is usual in
high-energy physics. The conformal extension $\Cl(2,4)$ is discussed in
Sec.~\ref{sec:limits}.

\begin{figure}[t]
\centering
\begin{tikzpicture}[
 every node/.style={font=\footnotesize, align=center, draw,
 rounded corners=2pt, minimum height=0.55cm,
 line width=0.4pt, inner sep=3pt},
 g0/.style={fill=blue!8, minimum width=1.5cm},
 g1/.style={fill=cyan!10, minimum width=6.0cm},
 g2/.style={fill=green!10, minimum width=9.0cm},
 g3/.style={fill=orange!10, minimum width=6.0cm},
 g4/.style={fill=red!10, minimum width=1.5cm}
]
 \node[g0] (g0) at (0, 4.4) {$\grade{\Cl(1,3)}{0}$, $\dim = 1$ \\ Mandelstam $p_i\!\cdot\!p_j$};
 \node[g1] (g1) at (0, 3.2) {$\grade{\Cl(1,3)}{1}$, $\dim = 4$ \\ four-momentum $p_i^\mu$};
 \node[g2] (g2) at (0, 2.0) {$\grade{\Cl(1,3)}{2}$, $\dim = 6$ \\ decay plane $p_i \wedge p_j$};
 \node[g3] (g3) at (0, 0.8) {$\grade{\Cl(1,3)}{3}$, $\dim = 4$ \\ oriented 3-volume $p_i \wedge p_j \wedge p_k$};
 \node[g4] (g4) at (0,-0.4) {$\grade{\Cl(1,3)}{4}$, $\dim = 1$ \\ pseudoscalar coefficient (CP-odd sign)};
 \draw[<->, dashed, black!60]
 (g0.east) to[out=0, in=0, looseness=2.0]
 node[draw=none, fill=none, midway, right=1pt, text=black!70] {$\star$}
 (g4.east);
 \draw[<->, dashed, black!60]
 (g1.east) to[out=0, in=0, looseness=1.5]
 node[draw=none, fill=none, midway, right=1pt, text=black!70] {$\star$}
 (g3.east);
\end{tikzpicture}
\caption{Grade ladder of $\Cl(1,3)$ with the physical meaning of each grade
for collider four-momenta. The dimensions $(1, 4, 6, 4, 1)$ are binomial
coefficients summing to $2^4 = 16$. The Hodge dual
$X^\star \equiv X\,I^{-1}$ pairs grade $k$ with grade $4 - k$ (dashed
arcs).}
\label{fig:grade-ladder}
\end{figure}

\subsection{The event multivector}
\label{sec:evmv-def}

Consider a reconstructed event with $N$ final-state detector objects
$i = 1,\dots,N$, with four-momenta $p_i \in \Cl(1,3)$ and discrete
object-type labels $f_i \in \mathbb F$. Its \emph{event multivector} is the
direct sum
\begin{equation}
 \evMV \;\equiv\;
 \bigoplus_{i=1}^{N}\;
 p_i \;\otimes\; |f_i\rangle
 \quad \in\quad
 \bigl[\Cl(1,3)\;\otimes\;\Vflav\bigr]^{\oplus N},
 \label{eq:evMV}
\end{equation}
where $\Vflav = \mathrm{span}_{\mathbb R}\{\,|f\rangle : f\in\mathbb F\,\}$
is a finite-dimensional real vector space, trivial under the Lorentz group,
with one basis vector per object type. The sum has one copy of
$\Cl(1,3)\otimes\Vflav$ per reconstructed object, so two objects of the same
type keep their individual momenta.

Equation~\eqref{eq:evMV} supplies only grade-one elements, a typed list of
four-vectors with no product between slots. The kinematic factors of all
slots, however, lie in the same algebra $\Cl(1,3)$ and can be multiplied
there. With $\pi_i:\evMV\mapsto p_i$ the projection onto slot $i$, a product
of $k$ distinct slots is labelled by the corresponding element of
$\Vflav^{\otimes k}$. The graded multivector of the introduction is the
indexed collection of these products,
\begin{equation*}
 \overline{\evMV} \;=\;
 \bigl(\,\pi_{i_1}\evMV \cdots \pi_{i_k}\evMV \;\otimes\;
 |f_{i_1}\rangle\!\otimes\!\cdots\!\otimes\!|f_{i_k}\rangle
 \,\bigr)_{1\le k\le 4,\; i_1\le\cdots\le i_k},
\end{equation*}
generated from the slots by the Clifford operations of Sec.~\ref{sec:ops}
acting on the $\Cl(1,3)$ factor, with the $\Vflav$ labels carried along.
Each entry keeps its value and its slot indices, so the masses
$m_i^2=p_i p_i$ appear next to the products of distinct slots. The entries
lie in $\Cl(1,3)\otimes\Vflav^{\otimes\bullet}$. Pairwise wedges of the $p_i$ give the $\binom{N}{2}$
decay-plane bivectors at grade two, such as the leptonic-$W$ plane
$p_\ell\!\wedge\!p_\nu$ of the $tWb$ example. Triple wedges give the
$\binom{N}{3}$ oriented three-volumes at grade three, and quadruple wedges
the $\binom{N}{4}$ pseudoscalar coefficients at grade four. Inner products
of the $p_i$ give the grade-zero Lorentz scalars (resonance masses,
Mandelstam variables, helicity-frame projections), and bivector exponentials
$\exp(B/2)$ give the rotors of Sec.~\ref{sec:symm}, which boost and rotate
the object. Figure~\ref{fig:per-particle-decomp} shows how all five grades
are filled from the grade-one input for $N=6$, and
Fig.~\ref{fig:event-as-multivector} summarises the
$\Cl(1,3)\!\otimes\!\Vflav$ factorisation seen by the network.

In this operational sense the event is a single geometric object, the graded
multivector generated from the per-object multivectors of
Eq.~\eqref{eq:evMV} by the operations of Sec.~\ref{sec:ops}. The higher
grades describe several particles jointly
($p\wedge p=0$, so a single particle contributes nothing above grade one).
Which grades a network receives \emph{explicitly} is an implementation
choice, and the choices differ in \emph{where} the higher grades are formed.
At one extreme the network receives
only the grade-one summands and builds these products in its equivariant
layers, as in the L-GATr-type configuration of Sec.~\ref{sec:tWb}. By
Theorem~\ref{lemma:cm-collapse} this choice loses no Lorentz-invariant
information, and Sec.~\ref{sec:ablation} measures what explicit higher-grade inputs add in a simple example. At the other extreme the whole inventory is supplied as separate
tokens. The pre-training interface of Sec.~\ref{sec:outlook} assumes this
form, which no implementation, ours included, yet supplies in full.
Throughout, $\grade{\evMV}{k}$ denotes the grade-$k$ sector of
$\overline{\evMV}$.

\subsection{Relation to symmetry-aware architectures}
\label{sec:baselines}

Symmetry-aware networks for collider data follow three established
approaches, which differ in where the symmetry is imposed. Our work concerns
the input to all three.

\emph{Invariance in the aggregation.} The PELICAN
family~\cite{Bogatskiy:2022czk, Bogatskiy:2023nnw} feeds a network the array
of inner products of the input momenta and aggregates them
permutation-equivariantly. Its output is exactly invariant without
equivariant layers. By Theorem~\ref{lemma:cm-collapse}, such a network cannot see the sign of the grade-four pseudoscalar, the one
invariant not reducible to inner products. LorentzNet~\cite{Gong:2022lye}
builds its messages from the same invariants but carries equivariant
four-vector channels through its layers, which places it between this family
and the next. Our analogue of this approach is the grade-zero-only row of
Table~\ref{tab:tWb-auc}. With everything else fixed, it has a lower area
under the ROC curve (AUC) than grade-$0\oplus1$ four-momenta, although it has
slightly more trainable parameters (Sec.~\ref{sec:ablation}).

\emph{Equivariance in the layers.} The Lorentz Group Network
(LGN)~\cite{Bogatskiy:2020tje}, Clifford-group equivariant
networks~\cite{Ruhe:2023rqc}, $\mathrm E(n)$-equivariant graph
networks~\cite{Satorras:2021aa} and the Lorentz-equivariant
geometric-algebra transformers~\cite{brehmer2023geometric, Spinner:2024hjm,
Brehmer:2024yqw} build the group action into the weights and carry covariant
quantities through the network. They are the closest to our work, which
shares their mechanism but specifies the content of the input tensor instead
of the layers. L-GATr also receives fixed reference objects, the time
direction as a vector and the beam either as the bivector $\gamma_{12}$ of
the transverse plane or as two beam-axis vectors~\cite{Spinner:2024hjm,
Brehmer:2024yqw}. Our beam bivector $\gamma_{03}$ is the Hodge dual of
$\gamma_{12}$, and Sec.~\ref{sec:ablation} compares it with vector reference
tokens in a network whose layers are held fixed.

\emph{Symmetry left to the data.} Deep-sets architectures on infrared-safe
observables~\cite{Komiske:2018cqr}, point-cloud convolutions with
augmentation~\cite{Qu:2019gqs} and attention with a pairwise-invariant
bias~\cite{Qu:2022mxj, Wu:2024thh} impose no exact group action and are
competitive on jet-level benchmarks. Whether a built-in symmetry helps is
therefore an empirical question, which Sec.~\ref{sec:ablation} answers for
our higher-grade inputs.

A fourth line of work addresses the combinatorial problem of our demonstration:
symmetry-preserving attention networks assign jets to their parent particles
directly~\cite{Fenton:2020woz, Shmakov:2021qdz}, without algebraic
structure. We do not benchmark against them because the two tasks have
different figures of merit. Those networks are scored on assignment accuracy
against a truth permutation, whereas our demonstration is a topology
classification with no assignment label. We vary only the input
representation and keep architecture, objective and data fixed.

\section{Physical meaning of multivector grades}
\label{sec:grades}

Each grade of the kinematic factor $\Cl(1,3)$ of the event multivector
$\evMV$ (Sec.~\ref{sec:evmv-def}) has a direct physical meaning for the
four-momenta of an event. The grade decomposition $\Cl(1,3) =
\bigoplus_{k=0}^{4}\grade{\Cl(1,3)}{k}$ applies to each slot of $\evMV$ in
Eq.~\eqref{eq:evMV}. Below we relate each grade to the observables of
Refs.~\cite{Boos:2008sdz, Dudko:2020qas} and state what it adds beyond grade
zero, given Theorem~\ref{lemma:cm-collapse}. The object-type factor $\Vflav$
(Eq.~\eqref{eq:Vflav}) and the missing-transverse-momentum pseudo-object
(Eq.~\eqref{eq:MET}) are treated in Sec.~\ref{sec:flav}.
Table~\ref{tab:mapping} maps $18$ variables onto the grades,
App.~\ref{app:tables-c1} gives the full $32$-row dictionary, and
Fig.~\ref{fig:per-particle-decomp} shows how the wedge products of an $N=6$
semi-leptonic $t\bar t$-like final state fill every grade of $\evMV$.

\begin{figure}[t]
\centering
\begin{tikzpicture}[
 every node/.style={font=\footnotesize, align=center},
 obj/.style={draw, rounded corners=2pt, fill=gray!10,
 minimum width=1.1cm, minimum height=0.55cm,
 line width=0.4pt, inner sep=2pt},
 g0/.style={draw, rounded corners=2pt, fill=blue!8,
 minimum width=4.2cm, minimum height=0.55cm,
 line width=0.4pt, inner sep=3pt},
 g1/.style={draw, rounded corners=2pt, fill=cyan!10,
 minimum width=4.2cm, minimum height=0.55cm,
 line width=0.4pt, inner sep=3pt},
 g2/.style={draw, rounded corners=2pt, fill=green!10,
 minimum width=4.2cm, minimum height=0.55cm,
 line width=0.4pt, inner sep=3pt},
 g3/.style={draw, rounded corners=2pt, fill=orange!10,
 minimum width=4.2cm, minimum height=0.55cm,
 line width=0.4pt, inner sep=3pt},
 g4/.style={draw, rounded corners=2pt, fill=red!10,
 minimum width=4.2cm, minimum height=0.55cm,
 line width=0.4pt, inner sep=3pt},
 cnt/.style={draw=none, font=\scriptsize\itshape, text=black!70},
 arrow/.style={-{Latex[length=1.4mm]}, line width=0.3pt, black!55}
]
\node[obj, minimum width=0.95cm] (l) at (-3.5, 4.7) {$\ell$\\\scriptsize$|f_\ell\rangle$};
\node[obj, minimum width=0.95cm] (nu) at (-2.1, 4.7) {$\nu$\\\scriptsize$|\MET\rangle$};
\node[obj, minimum width=0.95cm] (b) at (-0.7, 4.7) {$b$\\\scriptsize$|f_b\rangle$};
\node[obj, minimum width=0.95cm] (bb) at ( 0.7, 4.7) {$\bar b$\\\scriptsize$|f_b\rangle$};
\node[obj, minimum width=0.95cm] (j1) at ( 2.1, 4.7) {$j_1$\\\scriptsize$|f_j\rangle$};
\node[obj, minimum width=0.95cm] (j2) at ( 3.5, 4.7) {$j_2$\\\scriptsize$|f_j\rangle$};
\node[cnt, anchor=east] at (-3.95, 4.7) {reco};

\node[g4, minimum width=5.4cm] (g4) at (0, 3.6) {grade 4 ($\binom{6}{4}=15$):
 e.g.\ $p_\ell\!\wedge\! p_\nu\!\wedge\! p_b\!\wedge\! p_{\bar b}$};
\node[g3, minimum width=5.4cm] (g3) at (0, 2.7) {grade 3 ($\binom{6}{3}=20$):
 e.g.\ $p_\ell\!\wedge\! p_\nu\!\wedge\! p_b$};
\node[g2, minimum width=5.4cm] (g2) at (0, 1.8) {grade 2 ($\binom{6}{2}=15$):
 e.g.\ $p_\ell\!\wedge\! p_\nu$ ($W$ decay plane)};
\node[g1, minimum width=5.4cm] (g1) at (0, 0.9) {grade 1 ($N=6$):
 $\{p_\ell, p_\nu, p_b, p_{\bar b}, p_{j_1}, p_{j_2}\}$};
\node[g0, minimum width=5.4cm] (g0) at (0, 0.0) {grade 0: Lorentz scalars
 $\{p_i\!\cdot\! p_j,\ m_i^2,\ \hat s\}$};
\node[cnt, anchor=west] at (3.2, 3.6) {pseudoscalars (CP-odd signs)};
\node[cnt, anchor=west] at (3.2, 2.7) {trivectors};
\node[cnt, anchor=west] at (3.2, 1.8) {bivectors};
\node[cnt, anchor=west] at (3.2, 0.9) {four-momenta};
\node[cnt, anchor=west] at (3.2, 0.0) {Lorentz scalars};
\draw[decorate, decoration={brace, amplitude=3.5pt, mirror},
 line width=0.3pt, black!55]
 ($(l.south west)+(-0.05,-0.10)$) -- ($(j2.south east)+(0.05,-0.10)$);

\node[draw=gray!50, dashed, fit={(l)(j2)}, inner sep=4pt,
 label={[font=\scriptsize\itshape, gray]above:%
 $|f_i\rangle\in\Vflav,\ \dim=10$}]{};
\end{tikzpicture}
\caption{Per-particle assembly of the event multivector for a semi-leptonic
$t\bar t$-like final state $\{\ell, \nu, b, \bar b, j_1, j_2\}$ (lepton,
neutrino, two $b$-tagged and two light jets) with $N=6$ objects. Each object
carries a four-momentum $p_i$ (grade~1 in $\Cl(1,3)$) tensored with an
object-type tag $|f_i\rangle\in\Vflav$. The basis of $\Vflav$ is the ten
object types of Eq.~\eqref{eq:Vflav}, so $\dim\Vflav = 10$ does not depend
on~$N$. Grade $k$ of $\evMV$ collects all $\binom{N}{k}$ wedge products,
here $N=6$ four-vectors, $\binom{6}{2}=15$ bivectors (decay planes such as
$p_\ell\!\wedge\!p_\nu$ for the leptonic $W$), $\binom{6}{3}=20$ trivectors
and $\binom{6}{4}=15$ pseudoscalars. Grade zero holds the Lorentz scalars
$\{p_i\!\cdot\!p_j,\,m_i^2,\,\hat s\}$. Every object enters every grade, and
each higher grade shows one example. At reconstruction level the neutrino is
the $\MET$ pseudo-object of Eq.~\eqref{eq:MET}.}
\label{fig:per-particle-decomp}
\end{figure}

\paragraph{Grade 0: Lorentz scalars.}
The grade-zero sector $\grade{\evMV}{0}$ collects the Lorentz invariants of
the event four-momenta. All of them are polynomials in the inner products
$p_i\!\cdot\!p_j$, which include the on-shell masses $m_i^2 =
p_i\!\cdot\!p_i$ and give the invariant mass of any subset, $m_S^2 =
(\sum_{i\in S}p_i)^2$. The Mandelstam variables are subset masses as well,
$\hat s$ of the full final state and $\hat t$ and $\hat u$ of crossed
combinations with the initial partons. When the crossed invariants are used,
the initial partons are carried as two further slots of Eq.~\eqref{eq:evMV}
with their own $\Vflav$ labels. These invariants describe the
resonance structure of the event. The pair masses $m_{ij}$ peak at the
$W,Z,H$ masses and the three-body masses $m_{ijk}$ at $t$ and $\bar t$,
while the partonic $\sqrt{\hat s}$ controls threshold behaviour and
parton-luminosity weighting. The spin-analysing
projections $\cos\theta^*$ also belong here. They are the lepton angle in
the top rest frame and its optimal-basis variants, introduced for single-top
and $t\bar t$ spin correlations by Mahlon and Parke~\cite{Mahlon:1995zn,
Mahlon:1996pn, Mahlon:1999gz} and used in the input recipe of
Ref.~\cite{Boos:2008sdz}. Like the Collins--Soper angle $\cos\theta_{CS}$,
they are formed by a rotor sandwich $X\mapsto R\,X\,\reverse R$ followed by
an inner product, and so are Lorentz scalars in $\grade{\evMV}{0}$.

Hence $\grade{\evMV}{0}$ contains every Lorentz-invariant scalar input of
the recipes of Refs.~\cite{Boos:2008sdz, Dudko:2020qas}. The recipe of
Ref.~\cite{Dudko:2020qas} uses the pairwise inner products and $\hat s$
directly, and the high-level reconstructed masses of
Ref.~\cite{Boos:2008sdz} are square roots of linear combinations of the
same set (rows~1--4
of Table~\ref{tab:mapping}).

\begin{theorem}[Cayley--Menger collapse, Lorentzian extension]
\label{lemma:cm-collapse}
Let $p_1,\dots,p_N\in\Cl(1,3)$ be grade-one elements and let
$X_1,\dots,X_n$ be homogeneous-grade multivectors, each an outer product
of the $\{p_i\}$. Write $\lambda_{ijkl}$ for the pseudoscalar coefficient of
$\grade{p_i p_j p_k p_l}{4} = \lambda_{ijkl}\,I$
(Eq.~\eqref{eq:pseudoscalar}). Then
every Lorentz-invariant scalar $\grade{X_1 X_2\cdots X_n}{0}$ is a
polynomial in the inner products $\{p_i\!\cdot\!p_j\}_{i,j=1}^{N}$ and
the coefficients $\lambda_{ijkl}$, of degree at most one in the
$\lambda$'s jointly. It has the form
$P + \sum_{i<j<k<l}\lambda_{ijkl}\,Q^{ijkl}$ with $P$ and $Q^{ijkl}$
polynomials in the inner products alone, because every product
$\lambda_{ijkl}\lambda_{i'j'k'l'}$ is itself a polynomial in the inner
products. In particular the reversion contraction of any $r$-blade,
$r=2,3,4$, equals the $r\times r$ Gram determinant,
$\grade{P_r\reverse{P_r}}{0}=\det\bigl(p_{i_a}\!\cdot\!p_{i_b}\bigr)$
(Eqs.~\eqref{eq:gram2}, \eqref{eq:trivnorm}, \eqref{eq:gramR}).
\end{theorem}

\noindent\emph{Corollary.} The Gram matrix determines every
Lorentz-invariant scalar up to one bit. Under a spatial inversion all
$\lambda_{ijkl}$ change sign together while the inner products are
unchanged. Any ratio $\lambda_{T}/\lambda_{T'}$ is fixed by the Gram matrix
through
$\lambda_{T}\lambda_{T'}=-\det\bigl(p_{a}\!\cdot\!p_{b}\bigr)_{a\in T,\,b\in T'}$.
The single common orientation sign is therefore the only invariant content
not reducible to $\{p_i\!\cdot\!p_j\}$.

The proof is in App.~\ref{app:derivations}. The polynomial part is the First
Fundamental Theorem of invariant theory for the orthogonal
group~\cite[Ch.~II]{Weyl:1939}, specialised to $\mathrm{SO}(1,3)$. The
Gram-determinant identity is the Lorentzian form of the Euclidean
Cayley--Menger identity~\cite{Cayley:1841, Menger:1928}, used for
$(n,1)$-simplex realisability in Ref.~\cite[App.~A]{Tate:2011rm}. The
Clifford-algebra formulation adds that the only non-polynomial invariant,
the orientation bit, occupies a single channel of the event
multivector~$\evMV$. Section~\ref{sec:ablation} measures what the theorem implies for the input
layer.

The theorem restricts Lorentz-invariant scalars only, and grades two and three
remain an important part of the representation. A bivector $p_i\wedge p_j$ and
a trivector $p_i\wedge p_j\wedge p_k$ are covariant objects with six and four
components, which are not functions of the inner products. Each describes a
two- or three-particle subsystem as one geometric object, a decay plane with
the relative boost of its participants or an oriented decay volume. In the
scalar description this geometry enters only through Gram determinants of
several inner products. Grades two and three therefore give compact
information on the decay structure of the event, in the covariant form that an
equivariant network processes.

\paragraph{Grade 1: four-vectors.}
The grade-one sector $\grade{\evMV}{1}$ is spanned by the four-momenta $p_i$
and their sums (Fig.~\ref{fig:per-particle-decomp}, second row from the
bottom), including the total event four-momentum $p_{\text{tot}}=\sum_i
p_i$. The transverse momentum $p_T^f$, the rapidity $y^f$ and the azimuth
$\phi^f$ of object $f$ are not Lorentz scalars. They are components of $p_f$
with respect to a fixed beam-axis reference vector $\hat n_z$ and a fixed
transverse plane,
\begin{equation}
 p_f \;=\; (E^f,\, p_x^f,\, p_y^f,\, p_z^f)
 \;=\; (M_T^f \cosh y^f,\; p_T^f \cos\phi^f,\; p_T^f \sin\phi^f,\;
 M_T^f \sinh y^f),
 \label{eq:lhcparam}
\end{equation}
with $M_T^f \equiv \sqrt{m_f^2 + (p_T^f)^2}$. For massless objects the
rapidity $y$ equals the pseudorapidity $\eta=-\ln\tan(\theta/2)$, the
standard collider input when the object mass is small or unmeasured.

The lab-frame coordinates $(p_T,\eta,\phi)$ separate the content on which
longitudinal boosts act, $(M_T, y)$, from the content on which azimuthal
rotations act, $(p_T, \phi)$. The algebraic split
$\Cl(1,3)\cong\Cl(1,1)_z\,\hat\otimes\,\Cl(0,2)_\perp$ of
Eq.~\eqref{eq:22split} in Sec.~\ref{sec:symm} realises the same separation,
where $\hat\otimes$ is the $\mathbb Z_2$-graded tensor product of Clifford
algebras.

The standard event-level variables are functions of grade-one elements. Per
object they are $E,\,p_x,\,p_y,\,p_z,\,p_T,\,\eta,\, \phi,\,y,\,M$. The
event sums are $H_T=\sum_{\text{jets}}p_T$ and $S_T=H_T+\sum_\ell
p_T^\ell+\MET$, the pair variables are $\Delta\eta_{ij}, \Delta\phi_{ij},
\Delta R_{ij}=\sqrt{\Delta\eta_{ij}^2+\Delta\phi_{ij}^2}$ and $m_{ij}$ (the
last through grade zero), and the helicity-frame angle $\theta^*_\ell$
follows from the rotor sandwich of Eq.~\eqref{eq:rotor-sandwich}. The
transverse mass $m_T$ is invariant only under the LHC residual subgroup
$\Hlhc$ of longitudinal boosts and azimuthal rotations
(Eq.~\eqref{eq:Hlhc}). For network input we embed the azimuthal angle as
$(\cos\phi,\sin\phi)$ instead of the raw value, which jumps at
$\phi=\pm\pi$.

The sector $\grade{\evMV}{1}$ is the four-vector input of
Refs.~\cite{Dudko:2020qas, Brehmer:2024yqw}, here one homogeneous component
of the graded multivector $\evMV$. The geometric product mixes it with the
higher grades, and the rotor sandwich preserves each grade, so a
$\Spin^+(1,3)$-equivariant network acts on all grades without combining
channels by hand.

\paragraph{Grade 2: bivectors and decay planes.}
The grade-two sector $\grade{\Cl(1,3)}{2}$ is six-dimensional. In the
orthonormal basis $\{\gamma_\mu\}_{\mu=0,1,2,3}$ it splits into three boost
bivectors $\gamma_{0i}$ ($i=1,2,3$) with $\gamma_{0i}^2=+1$ and three
rotation bivectors $\gamma_{ij}$ ($1\le i<j\le 3$) with $\gamma_{ij}^2=-1$.
The two-blade $B_{ij}=p_i\wedge p_j$ has square
$B_{ij}^2=(p_i\!\cdot\!p_j)^2-p_i^2 p_j^2$ (Eq.~\eqref{eq:gram2}). For two
future-directed time-like momenta $B_{ij}\reverse{B_{ij}}\le 0$, so
$|B_{ij}|^2\equiv-B_{ij}\reverse{B_{ij}}$ is a non-negative Lorentz scalar
(Eq.~\eqref{eq:bivnorm}).

Physically, a unit two-blade $\hat B = B/|B|$ is the oriented two-plane in
Minkowski spacetime that contains the rest frames of its two participants.
For two on-shell time-like momenta it is also the generator of the boost
between these rest frames, with rapidity
$\alpha=\mathrm{arccosh}(p_i\!\cdot\!p_j/\sqrt{p_i^2 p_j^2})$
(Ref.~\cite[\S~5.107, eq.~5.125]{Doran:2003}). The rotor sandwich $B\mapsto
R\,B\,\reverse R$ keeps the blade form, so an equivariant network can carry
the orientation of a decay plane as a bivector. Examples are
$B_W=p_\ell\wedge p_\nu$ and $B_t=(p_\ell+p_\nu)\wedge p_b$, the oriented
decay planes of the leptonic $W$ and the leptonic $t$~\cite{Boos:2023kpp}.

The Bernreuther basis $\{\hat k,\hat r,\hat
n\}$~\cite{Bernreuther:2013aga, Bernreuther:2015yna} decomposes the
$t\bar t$ spin-density matrix into nine coefficients
$C_{ij}=-9\langle(\hat n_i\!\cdot\!\hat\ell_+)(\hat
n_j\!\cdot\!\hat\ell_-)\rangle$, averaged over events, with the lepton
directions $\hat\ell_\pm$ taken in the top and antitop rest frames. In each
event the factors $\hat n_i\!\cdot\!\hat\ell_\pm$ are normalised inner
products of a lepton four-momentum with axis four-vectors built from the top
momenta and the beam direction. Mutual projections of decay planes are
grade-zero quantities of the same kind. The identity
\begin{equation}
 \bigl\langle (p_a\wedge p_b)\,\reverse{(p_c\wedge p_d)}\bigr\rangle_0
 \;=\;
 (p_a\!\cdot\!p_c)(p_b\!\cdot\!p_d)
 \,-\,
 (p_a\!\cdot\!p_d)(p_b\!\cdot\!p_c)
 \label{eq:bivinner}
\end{equation}
places the projections $\grade{B_a\reverse{B_b}}{0}$ in the grade-zero sector of
Refs.~\cite{Boos:2008sdz, Dudko:2020qas}. The six covariant components of
$B_a$ do \emph{not} reduce to this set, and the multivector channels of the
equivariant networks of Refs.~\cite{Brehmer:2024yqw, Spinner:2024hjm} carry
them as they are. The azimuthal opening $\Delta\phi_{ij}$ extracts the
$\gamma_{12}$ projection of $B_{ij}$, and the rapidity difference $\Delta
y_{ij}$ the $\gamma_{03}$ projection. The distance $\Delta
R_{ij}=\sqrt{\Delta\eta_{ij}^2+\Delta\phi_{ij}^2}$ is a function of the same
projections, invariant under $\Hlhc$ but not under the full Lorentz group.
For massive objects this holds with the rapidity $y$ in place of $\eta$.

In QCD production from the unpolarised LHC beams the top quarks are almost
unpolarised, and their spins appear in correlations between production and
decay and between different particles~\cite{Bernreuther:2013aga,
Bernreuther:2015yna}. Electroweak single-top production gives strongly
polarised top quarks~\cite{Mahlon:1999gz}. Grade two carries both as covariant six-component objects, each combining
the orientation of a decay plane with the rapidity of the relative boost of
its participants (Sec.~\ref{sec:outlook}).

\paragraph{Event-shape variables.}
Sphericity $S$, aplanarity $A$, thrust $T$ and the Fox--Wolfram moments
$H_\ell$ are $\SOgroup(3)_{\rm lab}$-invariant, not Lorentz-invariant.
Sphericity is built from the \emph{symmetric} rank-2 tensor $S^{ij}=\sum_k
p_k^i p_k^j$, which lies in the orthogonal complement of
$\bigwedge^2(\mathbb R^{1,3})$ in the space of rank-2 tensors and is not an
element of the multivector representation. The jet-substructure variables
$n$-subjettiness $\tau_N$~\cite{Thaler:2010tr} and soft-drop mass
$m_{\text{SD}}$~\cite{Larkoski:2014wba} are defined through jet axes and
grooming, not through $S^{ij}$. In its $e^+e^-$ definition with polar angles,
$\tau_N$ is invariant under $\SOgroup(3)_{\rm lab}$, and the soft-drop mass
is invariant under $\Hlhc$. All these variables can still enter as
auxiliary scalar inputs at the grade-zero projection step, outside
$\overline{\evMV}$ (rows~19 and 21--25 of Table~\ref{tab:full-dict}).

\paragraph{Grade 3: oriented three-volumes.}
\begingroup\loosepar
The grade-three sector $\grade{\Cl(1,3)}{3}$ is four-dimensional and spanned
by $\gamma_{012},\,\gamma_{013},\, \gamma_{023},\,\gamma_{123}$. The Hodge
dual pairs each basis trivector with a basis vector,
$(\gamma_{abc})^\star=\mathrm{sgn}(\pi_{dabc})\,\gamma^d$, where
$\mathrm{sgn}(\pi_{dabc})$ is the sign of the permutation taking $(0,1,2,3)$
to $(d,a,b,c)$. A trivector $T\in\grade{\evMV}{3}$ is therefore, up to the
dual, a covariant four-component object. For three grade-one elements
\begin{equation}
 T_{ijk} \;=\; p_i\wedge p_j\wedge p_k,
 \qquad
 T_{ijk}\,\reverse{T_{ijk}}
 \;=\; \det G_3
 \;\equiv\; \det\bigl[\,p_a\!\cdot\!p_b\,\bigr]_{a,b\in\{i,j,k\}},
 \label{eq:trivnorm}
\end{equation}
in the convention of Theorem~\ref{lemma:cm-collapse}. By the
orthonormal-frame derivation of Eq.~\eqref{eq:gramR} the sign of this Gram
determinant is that of the metric induced on the three-plane of the momenta,
which admits at most one time-like direction. Hence $T\reverse T=\det G_3>0$
when the three-plane contains one time-like direction (induced signature
$(1,2)$), the dominant case for collider triples. A space-like three-plane
(induced signature $(0,3)$) gives $T\reverse T=\det G_3<0$, and a linearly
dependent triple gives $T\reverse T=\det G_3=0$.\par\endgroup

The trivector is the oriented three-volume in spacetime spanned by the three
four-momenta. The surface $\det G_3=0$ is the kinematic edge at which the
three momenta become linearly dependent, the boundary of a Dalitz-like
phase-space region. The Hodge dual $T^\star_{ijk}$ is the four-vector
orthogonal to the three-plane, the $\epsilon$-contraction
$T^{\star\,\mu}=\epsilon^{\mu\nu\rho\sigma}\,p_{i\nu}p_{j\rho}p_{k\sigma}$.
It is the covariant counterpart of the triple product $\vec
p_i\!\cdot\!(\vec p_j\!\times\!\vec p_k)$ of three-momenta.

Grade three corresponds to the T-odd and CP-odd triple products of
Refs.~\cite{Atwood:2000tu, Bernreuther:2015yna}, built there from
charge-conjugate pairs of momenta. The triple product is
P-odd, and T-odd in the motion-reversal convention of those references.
App.~\ref{app:tables-c2-disc} compares it with the covariant T used for
grade four below. In the lab or partonic-CM frame $\vec p_i\!\cdot\!(\vec
p_j\!\times\!\vec p_k)=T^{123}_{ijk}$ is one component of the trivector. The
other three, $T^{012},T^{013},T^{023}$, involve the time direction. Under
$\Spin^+(1,3)$ the four components transform, through the dual, as a Lorentz
vector, so an equivariant network sees the full four-vector
$T^\star_{ijk}\in\grade{\Cl(1,3)}{1}$ instead of the single projection
$T^{123}_{ijk}$.

\paragraph{Grade 4: pseudoscalars and the CP-odd sign.}
The grade-four component $\grade{\evMV}{4}=\lambda\,I$ is one-dimensional,
with $I=\gamma_0\gamma_1\gamma_2\gamma_3$ and $I^2=-1$ (the single
$\binom{4}{4}=1$ pseudoscalar box at the bottom of
Fig.~\ref{fig:grade-ladder}). For four grade-one elements
\begin{equation}
 P_{ijkl} \;=\; p_i\wedge p_j\wedge p_k\wedge p_l
 \;=\; \lambda_{ijkl}\,I,
 \qquad
 \lambda_{ijkl} \;=\;
 \epsilon_{\mu\nu\rho\sigma}\,p_i^\mu p_j^\nu p_k^\rho p_l^\sigma.
 \label{eq:pseudoscalar}
\end{equation}
Since $P_{ijkl}\reverse{P_{ijkl}}=-\lambda^2_{ijkl}$, Eq.~\eqref{eq:gramR}
for $r=4$ gives the square of the coefficient as minus the four-point Gram
determinant,
$\lambda^2_{ijkl}=-\det[\,p_a\!\cdot\!p_b\,]_{a,b\in\{i,j,k,l\}}$. The
pseudoscalar vanishes when fewer than four of the four-momenta are linearly
independent. This is the case for every final state with three or fewer
reconstructed objects when the initial state enters only through its total
momentum, which equals the sum of the final-state momenta. For $\ge 4$
linearly independent participants $\lambda\ne 0$ generically.

Physically, $\lambda$ is the signed four-volume spanned by the four momenta
in Minkowski spacetime. Its sign is a CP-odd one-bit observable, the only
Lorentz invariant \emph{inaccessible} to a classifier that sees
only inner products (Theorem~\ref{lemma:cm-collapse}). The pseudoscalar is
invariant under $\Spin^+(1,3)$, where $\det\Lambda=+1$. Under the full
$\mathrm{O}(1,3)$ it picks up a sign, $P\to(\det\Lambda)P$, so the parity
transformation $P$ flips $\lambda\to-\lambda$. The covariant time reversal
$\Lambda_T=\mathrm{diag}(-1,+1,+1,+1)$ also flips $\lambda\to-\lambda$,
because each non-zero term of the $\epsilon$-contraction in $\lambda$
contains a single $p^0$ factor. Charge conjugation $C$ leaves the momenta
unchanged and exchanges particles with antiparticles, so it permutes the
slots of $\lambda$ and multiplies it by the sign of that permutation. CP
analyses choose the four momenta so that $C$ swaps them in pairs, for
example $\epsilon(p_b,p_{\bar b},p_{\ell^+},p_{\ell^-})$ in dileptonic
$t\bar t$ decays. For such an ordering $\lambda$ is P-odd, T-odd and C-even,
hence CP-odd, and the CP-odd statements below refer to it. In the
Standard Model at tree level $\langle\mathrm{sgn}\,\lambda\rangle=0$ for
CP-symmetric configurations. A deviation signals CP violation only when the
contraction order is fixed by a physical rule and the sign is compared
between charge-conjugate samples (Sec.~\ref{sec:practical}).

The sign asymmetries of Refs.~\cite{Atwood:2000tu, Bernreuther:2015yna}
belong to grade four. For a chosen set of four momenta the sign of the
Bernreuther observable $T_5$ reduces to $\mathrm{sgn}\,\grade{p_a p_b p_c
p_d}{4}$, and the asymmetries $B_3$ and $\mathcal A_{CP}$ integrate the same
channel over event populations. The signed pseudoscalar has been used as a
CP-odd observable since Refs.~\cite{Atwood:2000tu, Bernreuther:2013aga,
Bernreuther:2015yna}, and our contribution is its formulation as an input
feature in the geometric-algebra framework. With the four covariant components of grade
three, the sign $\mathrm{sgn}\,\lambda$ completes the Lorentz-covariant
CP-odd kinematic content of the input layer.

\paragraph{Recovery of the prior input recipes.}
The recipes of Refs.~\cite{Boos:2008sdz, Dudko:2020qas} are contained in
$\grade{\evMV}{0}\oplus\grade{\evMV}{1}$, with Lorentz scalars and
helicity-frame angles in grade zero and lab-frame variables in grade one,
read in the beam frame through Eq.~\eqref{eq:lhcparam}. The
equivariant networks of Refs.~\cite{Bogatskiy:2020tje, Brehmer:2024yqw,
Spinner:2024hjm} also carry grades $\ge 2$ as raw covariant coefficients.
Table~\ref{tab:mapping} gives the mapping variable by variable.

\begin{table}[ht]
\centering
\small
\caption{Per-variable mapping of the input recipes~\cite{Boos:2008sdz,
Dudko:2020qas} onto the grades of the event multivector $\evMV$ and onto the
LHC residual subgroup
$\Hlhc=\SOgroup^+(1,1)_z\times\SOgroup(2)_\phi\subset\Spin^+(1,3)$ of
Eq.~\eqref{eq:Hlhc}. Rows~1--14 are standard collider variables (the union
of those recipes), cited only where an entry has a specific origin. Column ``Frame'' says where the variable is defined: Lorentz-invariant
(``inv.''), fixed in the beam, helicity, rest or lab frame, or a label with no
frame (``n.a.''). Column ``Inv.\ subgroup'' gives the largest group that
leaves the variable unchanged: ``Y''~stands for the full $\Spin^+(1,3)$, and
``$\Hlhc$''~for the LHC residual subgroup. The entry ``trivial'' marks labels
on which the Lorentz group does not act. The pseudoscalar sign of row~16 is
invariant under $\Spin^+(1,3)$ as well, and its entry ``N'' marks the one
invariant not reducible to inner products and masses
(Theorem~\ref{lemma:cm-collapse}). In this column ``n.a.''~marks variables
that are not elements of the multivector representation.}
\label{tab:mapping}
\footnotesize
\setlength{\tabcolsep}{4pt}
\begin{tabularx}{\linewidth}{@{}r l c l X@{}}
\hline\hline
\# & Variable & Grade & Frame & Inv.\ subgroup \\
\hline
 1 & $m_i^2$ & 0 & inv. & Y \\
 2 & $\hat s$ & 0 & inv. & Y \\
 3 & $\hat t$ & 0 & inv. & Y \\
 4 & $m_{ij}$ & 0 & inv. & Y \\
 5 & $p_T^i$ & 1 & beam & $\Hlhc$ \\
 6 & $\eta_i,\,y_i$ & 1 & beam & $\SOgroup(2)_\phi$ \\
 7 & $\phi_i$\textsuperscript{a} & 1 & beam & $\SOgroup^+(1,1)_z$ \\
 8 & $H_T,\,S_T$ & 1 & beam & $\Hlhc$ \\
   9 & $\cos\theta^*$ spin projections~\cite{Mahlon:1995zn, Mahlon:1996pn, Mahlon:1999gz} & 0 & hel. & Y \\
10 & $\Delta\phi_{ij},\,\Delta R_{ij}$ (with $y$) & 1$\to$0 & beam & $\Hlhc$ \\
11 & $C_{ij}$ spin correlation~\cite{Bernreuther:2013aga, Bernreuther:2015yna} & 0 & rest & $\Hlhc$ (event average) \\
12 & $\MET$ & 1$_{\text{(masked }E,\,p_z\text{)}}$ & beam & $\Hlhc$ \\
13 & $b$-tag, $\tau$-tag & $\Vflav$ & n.a. & trivial \\
14 & charge $Q_\ell$ & $\Vflav$ & n.a. & trivial \\
15 & $T^\star_{ijk}$ trivector dual (this work; cf.~\cite{Atwood:2000tu}) & 3$\to$1 & inv. & raw: $\SOgroup(3)_{\rm lab}$ for $T^{123}$; covariant: Y \\
16 & $\mathrm{sgn}\,\grade{p_1 p_2 p_3 p_4}{4}$ (this work; cf.~\cite{Atwood:2000tu, Bernreuther:2015yna}) & 4 & inv. & \textbf{N} (CP-odd 1-bit) \\
17 & sphericity $S$, aplanarity $A$ & n.a. & lab & n.a.\ (orth.\ complement) \\
18 & thrust $T$, Fox--Wolfram $H_\ell$ & n.a. & lab & n.a.\ (orth.\ complement) \\
\hline\hline
\end{tabularx}\\[2pt]
{\footnotesize\textsuperscript{a} For neural-network input we
recommend the embedding $(\cos\phi^f,\sin\phi^f)$ in place of the
raw angle.}
\end{table}

\section{Operations on the event multivector}
\label{sec:ops}

The grade decomposition of Sec.~\ref{sec:grades} fixes the content of the
event multivector $\evMV$ of Eq.~\eqref{eq:evMV}. We use seven operations
on $\Cl(1,3)$, each with a direct meaning for collider four-momenta and a
definite role in an equivariant network. They are the outer,
geometric and tensor products, the Hodge dual, the rotor sandwich, the
grade-bounded meet and join, and the bivector exponential.

\paragraph{Outer product $p\wedge q$.}
The outer product $p\wedge q = \tfrac12(pq - qp)$ of two grade-one
elements is the oriented Minkowski plane spanned by the two four-momenta,
an element of the six-dimensional sector $\grade{\Cl(1,3)}{2}$. By
Theorem~\ref{lemma:cm-collapse} its squared magnitude reduces to inner
products, $|p\wedge q|^2 = (p\!\cdot\!q)^2 - p^2 q^2$ in the convention of
Eq.~\eqref{eq:bivnorm}. Examples are the decay plane $p_\ell\wedge p_\nu$
of a leptonic $W$ boson and the plane $(p_\ell+p_\nu)\wedge p_b$ of its
parent top. Their mutual projections are grade-zero inputs of the same
kind as the per-event factors of the Bernreuther spin correlations
$C_{ij}$~\cite{Bernreuther:2013aga, Bernreuther:2015yna}
(Sec.~\ref{sec:grades}). An equivariant
network can take the raw two-blade as a covariant grade-two input.

\paragraph{Geometric product $pq = p\!\cdot\!q + p\wedge q$.}
The geometric product combines the symmetric inner and antisymmetric outer
parts into one bilinear operation, under which the sixteen-dimensional
algebra of Eq.~\eqref{eq:grade-decomp} closes. It is the basic
Lorentz-equivariant operation of the geometric-algebra transformers of
Refs.~\cite{brehmer2023geometric, Spinner:2024hjm, Brehmer:2024yqw}. Since
the pseudoscalar $I$ is itself $\Spin^+(1,3)$-invariant, every equivariant
bilinear map on multivector inputs is generated by grade projections of the
geometric product and of the pseudoscalar-dressed product
$\grade{x I y}{k}$. Closure under the geometric product is why the
equivariant input layer is built on $\Cl(1,3)$ and not on the Minkowski
vector space alone.

\paragraph{Hodge dual $X^\star \equiv X\, I^{-1}$.}
The Hodge dual pairs grade $k$ with grade $4-k$
(Fig.~\ref{fig:grade-ladder}, dashed arcs). It gives the trivector its
covariant four-vector form $T^\star_{ijk}$ (grade-3 entries of the
dictionary), which carries the Atwood--Soni-type CP-odd triple-product
features~\cite{Atwood:2000tu}. It enters the meet
(Eq.~\eqref{eq:meetjoin}), and applied to the grade-four part it returns
the pseudoscalar coefficient of $\grade{p_1 p_2 p_3 p_4}{4}$, the CP-odd
input feature of $\evMV$.

\paragraph{Rotor sandwich $X \mapsto R\, X\, \reverse R$.}
The proper orthochronous Lorentz group acts on a multivector through the
sandwich product of Eq.~\eqref{eq:rotor-sandwich}. The sandwich preserves
grade and realises any Lorentz transformation as one algebraic operation,
without a $4\times 4$ matrix acting on coordinates. The
laboratory-to-helicity transition $R_{\text{hel}}$, the
laboratory-to-Collins--Soper transition $R_{CS}$ and the Mahlon--Parke
optimal-spin-axis rotation~\cite{Mahlon:1999gz} are each one rotor, acting
in the same way on every grade of $\evMV$. The geometric-algebra equivariant
networks~\cite{brehmer2023geometric, Spinner:2024hjm, Brehmer:2024yqw} are
built around this operation, and the Lorentz Group
Network~\cite{Bogatskiy:2020tje} realises the same group action through
irreducible Lorentz-group representations.

\paragraph{Tensor product $\Cl(1,3) \otimes \Vflav$.}
Equation~\eqref{eq:evMV} attaches the discrete object-type label of each
detector object as a Lorentz-trivial tensor factor $\Vflav$, in the basis
of Eq.~\eqref{eq:Vflav}. This product is external to $\Cl(1,3)$. It
formalises the auxiliary scalar channel that the geometric-algebra
transformers of Ref.~\cite{Brehmer:2024yqw} add to each token. The
factorisation into a kinematic Clifford factor and an object-type factor
(Fig.~\ref{fig:event-as-multivector}) lets one architecture treat all
detector object classes alike, while the Lorentz action of
Eq.~\eqref{eq:spinact} stays confined to the $\Cl(1,3)$ side.
Section~\ref{sec:flav} compares it with algebraically richer alternatives.

\paragraph{Meet and join.}
The join and meet of two blades are
\begin{equation}
 A \cup B \;\equiv\; A \wedge B,
 \qquad
 A \cap B \;\equiv\;
 \bigl( A^\star \wedge B^\star \bigr)^{-\star},
 \label{eq:meetjoin}
\end{equation}
the join raising the grade and the meet lowering it. The join is the
smallest blade whose subspace contains those of both operands. Two linearly
independent four-momenta join into the decay-plane bivector
$p_i\wedge p_j$, and a four-momentum joined with a decay-plane bivector
gives the trivector of the corresponding top decay. The meet is the largest
blade whose subspace lies in both operands, their geometric intersection.
The double dual computes it as the join of the orthogonal complements of
$A$ and $B$, dualised back.

In $\Cl(1,3)$ both operations are \emph{grade-bounded} by the top grade
$4$. A join of operands whose grades add up to more than four vanishes, as
for three generic two-blades, and the meet of two distinct decay-plane
bivectors is already a scalar, covered by Theorem~\ref{lemma:cm-collapse}.
The conformal extension $\Cl(2,4)$ of Sec.~\ref{sec:limits}
(Eq.~\eqref{eq:conformal-embed}) lifts this bound. There two mass-shell
blades meet in a lower-grade blade that encodes the kinematic intersection
of the corresponding hyperboloids. Section~\ref{sec:limits} specifies this
construction at the representational level only. By the outermorphism
property $f(A\wedge B)=f(A)\wedge f(B)$ of any linear
map~$f$~\cite{Hestenes:1984}, a Lorentz transformation acts on every blade
through its action on the momenta and preserves the grade. Multi-blade
inputs can therefore enter the equivariant linear, geometric-product and
Lorentz-invariant attention layers of Sec.~\ref{sec:tWb} as covariant
objects.

The pairing tokens of Sec.~\ref{sec:tWb} apply Eq.~\eqref{eq:meetjoin} to
the two candidate top-decay trivectors. Their meet is a grade-two
covariant whose invariant norm indicates coplanarity of the two decay
volumes. Section~\ref{sec:ablation} measures whether supplying it
explicitly helps on the $tWb$ benchmark.

\paragraph{Bivector exponential $\exp(B/2)$.}
The map $B \mapsto R = \exp(B/2)$ is the smooth parameterisation of a
learnable rotor in equivariant networks~\cite{brehmer2023geometric,
Brehmer:2024yqw}. The bivector sign rule $\mathrm{sgn}(B^2)$ separates the
rotation generators ($B^2 < 0$ in our signature) from the boost generators
($B^2 > 0$). For instance, $B = \alpha\,\gamma_{03}$ ($B^2>0$)
exponentiates to a longitudinal boost of rapidity $\alpha$, and
$B = \phi\,\gamma_{12}$ ($B^2<0$) to an azimuthal rotation by $\phi$. One
bivector parameter describes boosts and rotations alike, without the
rotation/boost split of tensor formulations.

\section{Symmetries and equivariance}
\label{sec:symm}

\paragraph{Continuous spacetime symmetries.}
The full symmetry of the kinematic part of $\evMV$ is
$\Spin^+(1,3) \simeq \SLgroup(2,\mathbb C)$, acting by the sandwich product
\begin{equation}
 X \;\longmapsto\; R\,X\,\reverse R, \qquad
 R \in \Spin^+(1,3).
 \label{eq:spinact}
\end{equation}
Equation~\eqref{eq:spinact} preserves grade. It acts as the vector
representation on grade one, as the adjoint on grade two and, through the
Hodge dual, again as the vector representation on grade three. The
$\Spin^+(1,3)$-invariants of $\evMV$ are the grade-zero coefficients of
Sec.~\ref{sec:grades} and the grade-four coefficient, which is invariant
because $\det\Lambda=+1$ on the proper orthochronous group. Grades one to
three transform covariantly, and the equivariant networks of
Refs.~\cite{Bogatskiy:2020tje, brehmer2023geometric, Spinner:2024hjm,
Brehmer:2024yqw} process them in this form. Appendix~\ref{app:tables-c2} collects all
$30$ symmetries acting on $\evMV$ in four tables and classifies them by
enforcement (\emph{strict}, \emph{soft}, \emph{embedding},
\emph{conditional}, \emph{automatic}, \emph{out-of-scope}). Their rows are
cited below by number.

\paragraph{The LHC subgroup.}
Choosing a beam axis $\hat n_z$ and a transverse plane breaks
$\Spin^+(1,3)$ to the
longitudinal-boost~$\times$~azimuthal-rotation subgroup
\begin{equation}
 \Hlhc \;\equiv\;
 \SOgroup^+(1,1)_{z}\;\times\;\SOgroup(2)_\phi \;\subset\;
 \Spin^+(1,3)
 \quad\text{(via the rotor double cover of eq.~\eqref{eq:rotor-sandwich})},
 \label{eq:Hlhc}
\end{equation}
which acts on the beam-axis parameterisation $(p_T, y, \phi, M_T)$ of
Eq.~\eqref{eq:lhcparam}. Missing transverse momentum is compatible
only with $\Hlhc$, because the energy and longitudinal momentum of the
neutrinos are not measured (Sec.~\ref{sec:flav}). Both restrictions come from the collider,
not from the algebraic representation, and apply to every equivariant
network used at colliders (e.g.\ Ref.~\cite{Brehmer:2024yqw}).

\paragraph{Algebraic 2+2 split.}
The decomposition adapted to $\Hlhc$,
\begin{equation}
 \Cl(1,3) \;\cong\;
 \Cl(1,1)_z\;\hat\otimes\;\Cl(0,2)_\perp,
 \label{eq:22split}
\end{equation}
separates the longitudinal and transverse directions. The first factor
carries the longitudinal boosts $\Spin^+(1,1)$ and the second the azimuthal
rotations $\Spin^+(0,2)$, and $\hat\otimes$ is the graded tensor product
of Sec.~\ref{sec:grades}. The $(y,M_T)$ content of an event lies in
$\Cl(1,1)_z$ and the $(p_T,\phi)$ content in $\Cl(0,2)_\perp$.

\paragraph{Reference blades invariant under $\Hlhc$.}
The split of Eq.~\eqref{eq:22split} shows which constant algebraic objects
a network may carry as fixed reference tokens without breaking $\Hlhc$. The
two generators, the longitudinal boost $\gamma_{03}\equiv\gamma_0\gamma_3$
and the azimuthal rotation $\gamma_{12}$, act on a constant blade $B$ by
the commutator. The bivector $\gamma_{03}$ commutes with itself and, sharing
no index, with $\gamma_{12}$,
\begin{equation}
 [\gamma_{03},\,\gamma_{03}] = [\gamma_{12},\,\gamma_{03}] = 0,
 \label{eq:g03inv}
\end{equation}
so the beam bivector $\gamma_{03}$, the $(t,z)$ plane of the first factor,
is invariant under all of $\Hlhc$. The vectors $\gamma_0$ and $\gamma_3$
are not, since the longitudinal boost mixes them,
$[\gamma_{03},\gamma_0]=-2\gamma_3$ and
$[\gamma_{03},\gamma_3]=-2\gamma_0$. A network supplied with $\gamma_0$
and $\gamma_3$ separately keeps only the azimuthal-rotation symmetry. It
can form $E_i=\grade{T_i\gamma_0}{0}$ and
$p_{z,i}=-\grade{T_i\gamma_3}{0}$, from which $p_T$ follows, so it has the
ingredients of the $\Hlhc$ invariants but is not itself
$\Hlhc$-equivariant. A network supplied with $\gamma_{03}$ is equivariant
by construction. Section~\ref{sec:ablation} compares the two choices.

\paragraph{Discrete spacetime symmetries.}
Parity $P$ and time reversal $T$ act on $\Cl(1,3)$ as automorphisms that
lie outside the $\Spin^+(1,3)$ action. Parity is conjugation by
$\gamma_0$, the sandwich $X \mapsto \gamma_0\, X\, \gamma_0^{-1}$, which
sends $\gamma_0 \mapsto +\gamma_0$ and $\gamma_i \mapsto -\gamma_i$. Time
reversal flips $\gamma_0$ and is antilinear. Charge conjugation $C$ does
not act on $\Cl(1,3)$. It exchanges particle and antiparticle labels on
the object-type space $\Vflav$ (Sec.~\ref{sec:flav}). CP therefore acts on
$\Cl(1,3)\otimes \Vflav$ as the parity sandwich on the spacetime factor and
as the label exchange on the object-type factor.
Table~\ref{tab:symm-spacetime-disc} lists P, T, CP and CPT and their action
on the grades of $\evMV$.

\paragraph{Permutational symmetry.}
Reconstructed objects of the same type are interchangeable: swapping two
electrons, or two light jets, leaves the event unchanged. On $\evMV$ this
is the symmetry $S_{n_a}$ for each object-type label $a$ of $\Vflav$, which
relabels the per-object slots of Sec.~\ref{sec:evmv-def}. A network
enforces it by shared per-token weights in attention and a symmetric
read-out~\cite{Brehmer:2024yqw}, such as the mean pool of
Sec.~\ref{sec:tWb}.

\paragraph{Internal and gauge symmetries.}
The electromagnetic gauge group $\mathrm{U}(1)_{\text{em}}$ enters the
feature space as a charge bit on $\Vflav$ (Sec.~\ref{sec:flav}, row C.1),
so the network receives the charge explicitly. The electroweak group
$\mathrm{SU}(2)_L\times\mathrm{U}(1)_Y$, the colour group
$\mathrm{SU}(3)_C$ and the CKM and PMNS mixing matrices are not encoded.
The Monte-Carlo generator imposes them at the amplitude level (rows
C.2--C.5), and they reach $\evMV$ only through the structure of the
generated events. Lepton-flavour universality, realised on $\Vflav$ as an
$S_3$ permutation of $e,\mu,\tau$ tokens with the same lepton-type label,
can be imposed as an optional weight-tying of the corresponding channels.
Its breaking by the charged-lepton masses is recorded but not constrained,
and the implementation is left to the programme of Sec.~\ref{sec:outlook}.
Table~\ref{tab:symm-internal} lists the internal and permutational
symmetries in full.

\paragraph{Approximate and CP-channel symmetries.}
Several symmetries shape the architecture without being imposed as strict
equivariance (Table~\ref{tab:symm-approx-ml-cp}). Dilatation
$\mathrm{SO}(1,1)_D$, a global scaling of all four-momenta, is broken by
particle masses and selection thresholds and enters only softly, through
layer normalisation (row A.4 of Table~\ref{tab:symm-spacetime-cont}). The
pseudoscalar sign $\mathrm{sgn}\,\grade{p_1 p_2 p_3 p_4}{4}$ of
Sec.~\ref{sec:grades} (App.~\ref{app:tables-c2-approx}, row G.1) is one
CP-odd input channel of the pairing tokens of Sec.~\ref{sec:tWb}. In the
single-charge $\mu^-$ sample used there its symmetric distribution checks the
contraction order and does not test CP (Sec.~\ref{sec:practical}). The
particle tokens carry no positional encoding (row F.2), which preserves the
$S_{n_a}$ symmetry. The approximate scale invariance of the top-$p_T$ spectrum, a dilatation
group $\mathrm{SO}(1,1)$ broken logarithmically by the running coupling
(row A.4 of Table~\ref{tab:symm-spacetime-cont}), the chiral, heavy-quark and custodial symmetries (rows
E.1--E.3) and BSM CP and EDM channels (row G.3) lie outside the present
design. They are left to the programme of Sec.~\ref{sec:outlook}.

\section{Object-type space $\Vflav$}
\label{sec:flav}

The discrete reconstructed information (charged-lepton flavour, the
electric-charge bit, the $b$- and $\tau$-tag bits and the label of the
missing transverse momentum) does not transform under the Lorentz group. We
collect it in a finite-dimensional real vector space $\Vflav$, attached to
the kinematic factor by the tensor product of Eq.~\eqref{eq:evMV}. This
choice, Variant~A below, follows the auxiliary scalar channels of the
equivariant transformer~\cite{Brehmer:2024yqw}.

The rotor group $\Spin^+(1,3)$ acts on $\evMV$ only through the $\Cl(1,3)$
factor, by the sandwich product of Eq.~\eqref{eq:spinact}, and leaves the
labels $|f_i\rangle$ unchanged, since particle identity does not depend on
the inertial frame. The same global rotor $R$ boosts every $p_i$ and, by the
outermorphism property of Sec.~\ref{sec:ops}, every multi-particle blade
$p_i\!\wedge\!p_j\!\wedge\!\cdots$ built from them. Momenta of different
object types are combined by the products of Sec.~\ref{sec:ops}, for example
in the leptonic-$W$ plane $p_\ell\!\wedge\!p_\nu$ (an $e/\mu$ with $\MET$)
or in the top-decay trivector $p_\ell\!\wedge\!p_\nu\!\wedge\!p_b$.

\paragraph{Object-type basis.}
For the LHC final states considered here we use the basis
\begin{equation}
 \mathbb F \;=\;
 \bigl\{
 e^+,\, e^-,\, \mu^+,\, \mu^-,\,
 \tau_h^+,\, \tau_h^-,\,
 j,\, bj,\,
 \gamma,\,
 \MET
 \bigr\},
 \qquad
 \dim \Vflav \;=\; 10.
 \label{eq:Vflav}
\end{equation}
The basis covers the single-lepton-plus-jets signature of the $tWb$
demonstration and the other top-quark signatures of the HiGEN (Hierarchical
Geometric Event Network) programme. It contains the standard reconstructed
object types of ATLAS and CMS top-quark analyses at the operating point of
Refs.~\cite{Boos:2023kpp, Boos:2020rqy, Baskakov:2019rbq}. Nothing in
Eqs.~\eqref{eq:evMV}--\eqref{eq:spinact} fixes $\mathbb F$, and further tags
($c$-tag, boosted $W/Z/H/t$ substructure tags, forward-tagged jets) enlarge
$\mathbb F$ without changing the rest of the formalism.

The basis of Eq.~\eqref{eq:Vflav} is a reconstruction-level vocabulary. It
contains missing transverse energy in place of a neutrino, and it maps both
light jets to $j$ and both $b$-jets to $bj$ without separating quark from
antiquark. The parton-level demonstration of Sec.~\ref{sec:tWb} uses instead
the six-element truth-level basis $\{\mu^-,\bar\nu_\mu,u,\bar d,b,\bar b\}$,
in which the neutrino is an explicit four-momentum and the quark charge is
known. The missing-energy encoding of Eq.~\eqref{eq:MET} and the
flavour-tagging remarks below therefore do not enter the demonstration.

\begin{figure}[ht]
\centering
\begin{tikzpicture}[
 every node/.style={font=\footnotesize},
 block/.style={draw, rounded corners=2pt, minimum width=3.4cm,
 minimum height=0.55cm, align=center, line width=0.4pt},
 g0/.style={block, fill=blue!8},
 g1/.style={block, fill=cyan!10},
 g2/.style={block, fill=green!10},
 g3/.style={block, fill=yellow!12},
 g4/.style={block, fill=red!10},
 vf/.style={block, fill=gray!12, minimum width=4.0cm},
 arrow/.style={-{Latex[length=1.6mm]}, line width=0.4pt}
]
\node[align=center, text width=8.4cm] (root) at (-0.1, 0.75)
 {\large $\evMV$ (eq.~\ref{eq:evMV}) and its products $\overline{\evMV}$
  (Sec.~\ref{sec:evmv-def})};
\node[vf] (vfb) at (-3.6, -1.70) {object-type factor $\Vflav$\\
 $|f\rangle$, $\dim = 10$};
\node[g0] (g0b) at (3.5, -0.65) {grade 0: Lorentz scalars\\
 $\{p_i\!\cdot\!p_j,\, m_i^2,\, \hat s\}$};
\node[g1] (g1b) at (3.5, -1.70) {grade 1: $N$ four-vectors\\
 $\{p_i\}$};
\node[g2] (g2b) at (3.5, -2.75) {grade 2: $\binom{N}{2}$ bivectors\\
 $\{p_i\wedge p_j\}$};
\node[g3] (g3b) at (3.5, -3.80) {grade 3: $\binom{N}{3}$ trivectors\\
 $\{p_i\wedge p_j\wedge p_k\}$};
\node[g4] (g4b) at (3.5, -4.85) {grade 4: $\binom{N}{4}$ pseudoscalars\\
 $\{p_i\wedge p_j\wedge p_k\wedge p_l\}$};
\draw[arrow] (root.south) -- ++(0,-0.2) -| (vfb.north);
\draw[arrow] (root.south) -- ++(0,-0.2) -| (g0b.west);
\draw[arrow] (root.south) -- ++(0,-0.2) -| (g1b.west);
\draw[arrow] (root.south) -- ++(0,-0.2) -| (g2b.west);
\draw[arrow] (root.south) -- ++(0,-0.2) -| (g3b.west);
\draw[arrow] (root.south) -- ++(0,-0.2) -| (g4b.west);
\node[draw=gray!50, dashed, fit=(g0b)(g4b), inner sep=5pt,
 label={[font=\footnotesize\itshape, align=center, yshift=-1mm]below:
 Lorentz-equivariant kinematic factor $\Cl(1,3)$}] {};
\end{tikzpicture}
\caption{Decomposition of the event multivector $\evMV$ into the kinematic
Clifford factor $\Cl(1,3)$ and the Lorentz-trivial object-type factor
$\Vflav$. The object-type factor on the left carries one entry of the basis
of Eq.~\eqref{eq:Vflav} per detector token. The per-grade counts on the
right are for an event with $N$ final-state objects. The parton-level $tWb$
demonstration of Sec.~\ref{sec:tWb} has $N=6$, and reconstructed
single-lepton-plus-jets events with ISR jets typically have
$N\!\lesssim\!8$. Grades ${\geq}2$ arise from products of the grade-one
input (Sec.~\ref{sec:evmv-def}), not from Eq.~\eqref{eq:evMV} itself, and
their type labels lie in tensor powers of $\Vflav$
(Table~\ref{tab:symm-internal}). Section~\ref{sec:grades} gives the physical
meaning of each grade.}
\label{fig:event-as-multivector}
\end{figure}

\paragraph{Variant choice.}
We considered four ways to attach the object-type labels to the kinematic
factor. Variant~A, used throughout the paper, keeps the labels $|f\rangle$ in a
fixed real vector space, and its Lorentz factor commutes with any
relabelling of object types. Variant~B replaces the tensor product by the
extended Clifford algebra $\Cl(1,3+n_f)$, in which $n_f$ extra basis vectors
carry the object-type labels and the algebraic charge-conjugation
automorphism acts non-trivially on these flavour generators. It is the only
alternative with an algebraic advantage, namely that $C$ becomes an inner
automorphism of the algebra instead of an involution on a separate factor.

We reject Variant~B for three reasons. First, the extended algebra has
$2^{4+n_f}$ basis elements per token. For the vocabulary of
Eq.~\eqref{eq:Vflav} ($n_f=10$) the count is $16\,384$, against
$\dim\bigl(\Cl(1,3)\otimes\Vflav\bigr)=16\times 10=160$ for the tensor
product, a factor $102$. Second, and more seriously, the spin group
$\Spin(1,3+n_f)$ contains transformations that mix spacetime with
object-type directions. They correspond to no physical symmetry, so
equivariance under this group would impose a symmetry the data do not have.
Third, we know of no equivariant layers available at this dimensionality.

Variant~C is the multi-particle Clifford bundle $\Cl(1,3)^{\otimes N}$ of
Refs.~\cite{Doran:2003,Hestenes:1984}, with one copy of the algebra per
particle. Variant~D keeps the tensor product and gives $\Vflav$ a $\mathbb
Z_2$ grading by charge conjugation. A particle and its antiparticle then share
a $C$-even label component and differ by the sign of a $C$-odd component, so
that $C$ acts as multiplication by $\pm 1$. Variants~C and~D are not
implemented in this paper. With Variant~A, the particle and antiparticle
labels of Eq.~\eqref{eq:Vflav} are independent unit vectors, which $C$
exchanges. The network has separate weights for the two labels. Its responses
to an event and to the charge-conjugated event are therefore not related by
construction, and any such relation has to be learned from the data. In return
the input follows the ATLAS/CMS object conventions and the scalar channel of
Refs.~\cite{Spinner:2024hjm, Brehmer:2024yqw} directly.

\paragraph{The missing-energy pseudo-object.}
The missing transverse momentum is the imbalance of the visible
transverse momenta,
$\vec p_T^{\mathrm{miss}} \equiv -\sum_{i\in\text{vis}}\vec p_T^{\,i}$,
with components $(p_x^{\mathrm{miss}}, p_y^{\mathrm{miss}})$ and magnitude
$\MET \equiv |\vec p_T^{\mathrm{miss}}|$. We treat it as a pseudo-object
with grade-one component
\begin{equation}
 p_{\MET} \;=\; (0,\, p_x^{\mathrm{miss}},\, p_y^{\mathrm{miss}},\, 0),
 \qquad
 f_{\MET} \;=\; |\MET\rangle,
 \label{eq:MET}
\end{equation}
and a binary mask $m^{\mathrm{meas}}_{\MET} = (0,\,1,\,1,\,0)$ on the
grade-one components of the $\Cl(1,3)$ factor, which marks the energy and
the longitudinal momentum as unmeasured. The mask is propagated through the
downstream layers, so the space-like value $p_{\MET}^2 < 0$ is never used as
a physical quantity. Because the energy and longitudinal momentum of the
neutrino system are not measured, $p_{\MET}$ transforms as a four-momentum
only under the LHC subgroup $\Hlhc$ of Eq.~\eqref{eq:Hlhc}. We use the
partial four-momentum of Eq.~\eqref{eq:MET} because it keeps every token in
the same format. Table~\ref{tab:met-encodings} lists three alternatives.

\begin{table}[ht]
\centering
\footnotesize
\caption{Candidate algebraic encodings of $\MET$. The partial
four-momentum (row~1) is adopted.}
\label{tab:met-encodings}
\renewcommand{\arraystretch}{1.15}
\begin{tabularx}{\linewidth}{@{}>{\raggedright\arraybackslash}p{4.2cm}
 >{\raggedright\arraybackslash}p{2.0cm}
 >{\raggedright\arraybackslash}X
 >{\raggedright\arraybackslash}X@{}}
\hline\hline
\textbf{Encoding} & \textbf{Grade} & \textbf{Preserves} & \textbf{Cost} \\
\hline
Partial four-momentum (Eq.~\eqref{eq:MET}) & 1 (masked $E,p_z$)
 & per-token uniformity & explicit binary mask \\
Transverse two-blade $p_T^{\mathrm{miss}}\wedge\hat n_z$ & 2
 & $\SOgroup(2)_\phi$ orbit
 & breaks per-token uniformity \\
$p_z(\nu)$ from $W$-mass~\cite{Boos:2023kpp,Boos:2020rqy} & 1 (full)
 & full Lorentz orbit
 & quadratic; single-$\nu$ \\
$2{+}2$ split $\Cl(1,1)\!\otimes\!\Cl(0,2)$ & basis change
 & longitudinal/transverse symmetry of the LHC
 & non-trivial rewrite \\
\hline\hline
\end{tabularx}
\renewcommand{\arraystretch}{1.0}
\end{table}

\paragraph{Discrete symmetries on $\Vflav$.}
Of the discrete spacetime symmetries only charge conjugation acts on
$\Vflav$, as the $\mathbb Z_2$ exchange of particle and antiparticle labels
(row B.3 of Table~\ref{tab:symm-spacetime-disc}). Parity and time reversal
act on the kinematic factor alone (Sec.~\ref{sec:symm}). The permutational
symmetries on $\Vflav$ are rows D.1--D.5 of App.~\ref{app:tables-c2-int}.
The per-flavour permutation $\prod_a S_{n_a}$ of identical objects is
realised by set-attention, so the order of the slots in Eq.~\eqref{eq:evMV}
does not matter. The lepton-flavour-universality permutation $S_3$ among
$e,\mu,\tau$ tokens is a soft weight-tying option.

\paragraph{Information not encoded.}
The $b$-tag enters only as the binary slot $|bj\rangle\in \mathbb F$. The
impact-parameter and secondary-vertex information is not encoded, and a
vertex-level encoding is among the extensions of Sec.~\ref{sec:limits}.
Fermion spin states are not represented as Dirac spinors $\Psi\in\Cl(1,3)$.
Spin information enters only implicitly, through the second moment
$\evMV\otimes\evMV$ (the decay-plane bivectors and Bernreuther
correlations of Sec.~\ref{sec:grades}).

\section{Resonance-topology separation in $pp\to tWb$}
\label{sec:tWb}

The algebraic representation is applied here to a simple case, the
parton-level separation of single- and double-resonant contributions to the
$pp\to tWb$ final state. The case is an example of how the representation is
used. The network is a small single-task classifier and does not implement the
foundation model outlined in Sec.~\ref{sec:outlook}. A compact geometric-algebra network on
multivector input (\netGa) is compared with the reference network of
Ref.~\cite{Boos:2023kpp} (\netRef), which uses established high-level
features.

\paragraph{Physics target.}
The process $pp\to tWb$ is the lowest-multiplicity hadron-collider final
state in which one and two Breit--Wigner top resonances coexist in the
same matrix element. The associated single-top channel $pp\to tW^-\bar b$
has one resonant top, $t\to W^+ b$. The doubly-resonant $t\bar t$
component of the same six-particle final state has two, $t\to W^+ b$ and
$\bar t\to W^- \bar b$. The same separation problem, with more resonances
and a larger combinatorial burden, arises in triple-top
production~\cite{Boos:2021yat} and in four-top final states. There, diagram
classes with three and four top resonances share a final state with the
lower-multiplicity classes.
The $tWb$ system also enters the recent CMS~\cite{CMS:2025kzt} and
ATLAS~\cite{ATLAS:2026dbe, ATLAS:2026qre} observations of a cross-section
enhancement near the $t\bar t$ production threshold. The enhancement is
consistent with a colour-singlet quasi-bound toponium state, with cross
sections of $8.8^{+1.2}_{-1.4}$~pb and $9.3^{+1.4}_{-1.3}$~pb,
respectively. In those measurements the modelling of the gauge-invariant
$tWb$ matrix element used here, in particular of the $t\bar t/tW$ overlap
that the two truth classes below isolate, is one of the leading sources of
systematic uncertainty. Representative diagrams are shown in
Fig.~\ref{fig:tWb-feynman}(a),(b).

\begin{figure}[t]
\centering
\begin{tikzpicture}[
 every node/.style={font=\footnotesize},
 fermion/.style={postaction={decorate},
 decoration={markings, mark=at position 0.55 with
 {\arrow[line width=0.5pt]{Latex[length=1.5mm]}}},
 line width=0.5pt},
 fermionR/.style={postaction={decorate},
 decoration={markings, mark=at position 0.55 with
 {\arrowreversed[line width=0.5pt]{Latex[length=1.5mm]}}},
 line width=0.5pt},
 gluon/.style={decorate,
 decoration={coil, aspect=0.5, segment length=1.4mm,
 amplitude=0.7mm},
 line width=0.5pt},
 boson/.style={decorate,
 decoration={snake, segment length=1.6mm, amplitude=0.5mm},
 line width=0.5pt},
 res/.style={line width=1.0pt},
 vtx/.style={circle, fill=black, inner sep=1pt, minimum size=2pt},
 lbl/.style={font=\scriptsize, inner sep=1pt}
]
\begin{scope}[shift={(-4.2,0)}]
 \node[lbl, anchor=east] at (-2.6, 1.2) {$g$};
 \node[lbl, anchor=east] at (-2.6,-1.2) {$g$};
 \node[vtx] (Wgg) at (-1.4, 0.0) {};
 \node[vtx] (Wtt) at (-0.2, 0.6) {};
 \node[vtx] (Wbar) at ( 0.7,-0.4) {};
 \node[vtx] (WW) at ( 1.6,-1.2) {};
 \draw[gluon] (-2.6, 1.2) -- (Wgg);
 \draw[gluon] (-2.6,-1.2) -- (Wgg);
 \draw[gluon] (Wgg) -- (Wtt)
 node[lbl, midway, above=1pt, fill=white, inner sep=1pt] {$g^{\!*}$};
 \draw[fermion, res] (Wtt) -- (2.8, 1.4)
 node[lbl, anchor=west] {$t$};
 \draw[fermionR, res] (Wtt) -- (Wbar)
 node[lbl, midway, left=2pt, fill=white, inner sep=1pt] {$\bar t$};
 \draw[fermionR] (Wbar) -- (2.8, 0.0)
 node[lbl, anchor=west] {$\bar b$};
 \draw[boson] (Wbar) -- (WW)
 node[lbl, midway, right=2pt, fill=white, inner sep=1pt] {$W^-$};
 \draw[fermionR] (WW) -- (2.8,-0.6)
 node[lbl, anchor=west] {$\bar\nu_\mu$};
 \draw[fermion] (WW) -- (2.8,-1.7)
 node[lbl, anchor=west] {$\mu^-$};
 \draw[decorate, decoration={brace, amplitude=2pt},
 line width=0.3pt, black!55]
 (3.4, 0.0) -- (3.4,-1.7)
 node[lbl, midway, right=4pt, text=black!75]
 {$=m_t$};
 \node[font=\footnotesize] at (0.4,-2.3)
 {(a) double-resonant: $t\bar t$};
\end{scope}
\begin{scope}[shift={(3.0,0)}]
 \node[lbl, anchor=east] at (-2.6, 1.2) {$g$};
 \node[lbl, anchor=east] at (-2.6,-1.2) {$g$};
 \node[vtx] (Vb) at (-1.0,-1.2) {};
 \node[vtx] (Vw) at (-1.0, 0.0) {};
 \node[vtx] (Vt) at (-1.0, 1.2) {};
 \node[vtx] (Vlep) at ( 0.8, 0.0) {};
 \draw[gluon] (-2.6, 1.2) -- (Vt);
 \draw[gluon] (-2.6,-1.2) -- (Vb);
 \draw[fermionR] (Vb) -- (2.8,-1.4)
 node[lbl, anchor=west] {$\bar b$};
 \draw[fermion] (Vb) -- (Vw)
 node[lbl, midway, left=2pt, fill=white, inner sep=1pt] {$b^{\!*}$};
 \draw[fermion] (Vw) -- (Vt)
 node[lbl, midway, left=2pt, fill=white, inner sep=1pt] {$t^{\!*}$};
 \draw[fermion, res] (Vt) -- (2.8, 1.4)
 node[lbl, anchor=west] {$t$};
 \draw[boson] (Vw) -- (Vlep);
 \node[lbl, fill=white, inner sep=1pt] at (-0.1, 0.0) {$W^-$};
 \draw[fermion] (Vlep) -- (2.8, 0.4)
 node[lbl, anchor=west] {$\mu^-$};
 \draw[fermionR] (Vlep) -- (2.8,-0.4)
 node[lbl, anchor=west] {$\bar\nu_\mu$};
 \draw[decorate, decoration={brace, amplitude=2pt},
 line width=0.3pt, black!55]
 (3.4, 0.4) -- (3.4,-1.4)
 node[lbl, midway, right=4pt, text=black!75]
 {$\ne m_t$};
 \node[font=\footnotesize] at (0.4,-2.3)
 {(b) single-resonant: $tW$};
\end{scope}
\end{tikzpicture}
\caption{Two representative tree-level diagrams for the
$gg\to t\bar b\,\mu^-\bar\nu_\mu$ final state of the matrix-element
set of Refs.~\cite{Baskakov:2019rbq, Boos:2020rqy}. The braces on the
right mark the invariant mass of the $\bar b\,\mu^-\bar\nu_\mu$ system,
which distinguishes the two topologies.
(a)~Double-resonant ($t\bar t$): $\bar b$, $\mu^-$ and
$\bar\nu_\mu$ come from the same $\bar t$ line, so
$m_{\bar b\mu^-\bar\nu_\mu}\!=\!m_t$ on the Breit--Wigner pole.
(b)~Single-resonant ($tW$): the $\bar b$ comes from a
$g\!\to\!b\bar b$ splitting, not from a top decay, so
$m_{\bar b\mu^-\bar\nu_\mu}\!\ne\!m_t$.}
\label{fig:tWb-feynman}
\end{figure}

\paragraph{Truth definition of the two classes.}
The two classes are defined at the matrix-element level by a
diagram-removal scheme~\cite{Boos:2020rqy, Boos:2023kpp}. Class~A
(single-resonant) keeps only the electroweak tree-level $tWb$ diagrams
with one top resonance, $t\to W^+ b$. Class~B (double-resonant) keeps the
$t\bar t$ diagrams with two, $t\to W^+ b$ and $\bar t\to W^- \bar b$. A
third sample, generated with the full gauge-invariant matrix element
(both classes and their interference), is not used in training and serves
only as a check on the inclusive event population. Truth labels are taken
from the parton-level record of the matrix-element generator.

\paragraph{Generator and event selection.}
\label{para:gen}
Events are generated at $\sqrt s = 14$~TeV with the matrix-element
generator \textsc{CompHEP~4.5}~\cite{CompHEP:2004qpa} for the six-particle
final state $gg \to \mu^-\,\bar\nu_\mu\,u\,\bar d\,b\,\bar b$. Here the
leptonic decay $\bar t\to\bar b\,W^-\to\bar b\,\mu^-\bar\nu_\mu$, the
hadronic decay $t\to b\,W^+\to b\,u\bar d$ and associated single-top
production $gg\to tW^-\bar b$ ($t$ hadronic, $W^-\to\mu^-\bar\nu_\mu$) are
diagram classes of one matrix element. The single- and double-resonant
contributions and their interference are therefore generated together. No parton
shower or detector simulation is applied, which isolates the effect of the
input representation. The matrix-element scheme, parton distribution,
factorisation and renormalisation scales and the event selection
$p_T^{j_4}>10$~GeV are those of Ref.~\cite{Boos:2023kpp}.

\paragraph{Reference network.}
The reference network \netRef\ of Ref.~\cite{Boos:2023kpp} is a fully
connected feed-forward network on ${\sim}75$ high-level features that
combine the recipes of Refs.~\cite{Boos:2008sdz, Dudko:2020qas}. Most
features are logarithms: of the pairwise dot products
$\log\,p_i\!\cdot\!p_j$ of the reconstructed four-momenta, of the
transverse momenta $\log p_T$, of the reconstructed top and $W$ masses and
of $\hat s$, $\MET$ and $H_T$. The dot products are positive on the
kinematics of interest and, for light final-state partons, equivalent to
two-body invariant masses. A complete set of pseudorapidities and
helicity-frame cosines enters without a logarithm. Three hidden
layers of five hundred units each, with $L_2$-regularisation and dropout,
are trained with binary cross-entropy, inverse-class-frequency weighting
and full-batch Adam. The training uses the same generated samples as \netGa, about
270k $tW$ and 10M $t\bar t$ parton-level events passing
$p_T^{j_4}>10$~GeV. The published
checkpoint, with $5.4\!\times\!10^5$ trainable parameters, is used without
retraining. Reference~\cite{Boos:2023kpp} documents the training and the
complete list of input variables.

\paragraph{Geometric-algebra network.}
\netGa\ is a compact variant of L-GATr, a multivector-equivariant
transformer of the family of Ref.~\cite{Brehmer:2024yqw}. It takes the same
events as multivectors of $\Cl(1,3)\otimes\Vflav$. Each of the six
final-state partons is one token $T_i = \grade{T_i}{0}\oplus\grade{T_i}{1}$.
Its grade-one part $\grade{T_i}{1}=p_i^\mu\gamma_\mu$ is the four-momentum,
and $\grade{T_i}{0}$ carries a one-hot encoding of the object type
($\mu^-,\bar\nu_\mu,u,\bar d,b,\bar b$) and the electric-charge and $b$-tag
flags as scalar channels. Each of the three equivariant blocks has a
per-grade \emph{equivariant linear} layer that mixes channels within each
grade~\cite{Brehmer:2024yqw}, a \emph{geometric-product} layer and a
\emph{Lorentz-invariant attention} head. The geometric-product layer
contracts pairs of multivector channels and projects the result back by
grade-aware linear mixing, so that bivectors $p_i\wedge p_j$ and trivectors
$p_i\wedge p_j\wedge p_k$ enter as covariant objects of their grades. The
attention scores are the scalar inner products $\langle T_i, T_j\rangle_0$
of multivector keys and queries. Every operation either projects onto
invariant grade-zero scalars or acts covariantly on multivector channels,
commuting with the rotor sandwich applied to every token, so the network is
$\Spin^+(1,3)$-equivariant by construction. The read-out takes the grade-zero coefficients of the final
tokens, applies a permutation-invariant mean over the six tokens and feeds a
small dense head that gives the binary discriminant. With the
candidate-pairing tokens the network has $151\,489$ trainable parameters
against $539\,501$ for \netRef, a factor $3.6$ fewer, so network capacity
cannot explain an advantage of \netGa. Training uses AdamW with a cosine
learning-rate schedule (initial learning rate $10^{-3}$, weight decay
$10^{-4}$), mini-batch size $512$, gradient-norm clipping at $1.0$ and
binary cross-entropy loss. \netGa\ is trained on the same samples as
\netRef, of which 9.3M events enter training after the held-back split, with
a balanced random sampler that draws ${\sim}500$k events per epoch. Training
runs for $50$ epochs, a few hours per configuration on a single
consumer-class GPU. Every AUC is taken at the last epoch, so no model is
selected by the quantity being reported.

\paragraph{Event-level pairing tokens.}
The $tW$ and $t\bar t$ classes share the same six-particle final state
($\mu^-$, $\bar\nu_\mu$, two light jets $j_1,j_2$, and two $b$-tagged jets
$b_1,b_2$). At detector level the single- and double-resonant topologies
cannot be separated without assigning the two $b$-tagged jets to the
candidate top-quark resonances, the classical reconstruction problem of
associated $tWb$ production. At parton level the object-type labels
distinguish $b$ from $\bar b$ and fix the assignment. The pairing tokens
are nevertheless built symmetrically over both hypotheses, as at detector
level, where no such label exists. We add two \emph{event-level pairing
tokens}, one for each candidate $b$-assignment $a\in\{1,2\}$, which bring
grade-two and grade-three content of the multi-resonance topology to the
input layer (eight tokens in total). Their covariant channels carry the
candidate hadronic-top trivectors
$T_{\rm had}^{(a)}=p_{j_1}\wedge p_{j_2}\wedge p_{b_a}$ (grade three),
their Hodge duals $\star T_{\rm had}^{(a)}$ (grade one), and the meet
$T_{\rm had}^{(a)}\cap T_{\rm lep}^{(a)}$ of the two candidate top
trivectors, with
$T_{\rm lep}^{(a)}=p_\mu\wedge p_{\bar\nu}\wedge p_{b_{\bar a}}$ built from
the complementary $b$-jet. The meet is grade two, because the dual of a
grade-three blade is grade one and the wedge of two grade-one elements is
grade two. Both $b$-assignments thus enter the input together, so that the
equivariant attention head can weight them without an explicit
combinatorial step. We do not measure the assignment accuracy of that
head. The scalar channels (grade zero) carry the Breit--Wigner pulls
(Eq.~\eqref{eq:pairing-pull}) of the two candidate top-mass
reconstructions,
\begin{equation}
 s^{(a)}_{\rm had} \;\equiv\;
 \frac{(p_{j_1}+p_{j_2}+p_{b_a})^2 - m_t^2}{m_t\,\Gamma_t},
 \qquad
 s^{(a)}_{\rm lep} \;\equiv\;
 \frac{(p_\mu+p_{\bar\nu}+p_{b_{\bar a}})^2 - m_t^2}{m_t\,\Gamma_t},
 \qquad \bar a \equiv 3-a,
 \label{eq:pairing-pull}
\end{equation}
the symmetric and antisymmetric $Wb$ combinations (with
$W_{\rm had}\equiv p_{j_1}+p_{j_2}$ the hadronic $W$ candidate)
\begin{equation*}
 \sigma^{(\pm)}_{Wb} \;=\; \tfrac12\bigl[(W_{\rm had}+p_{b_1})^2
 \pm (W_{\rm had}+p_{b_2})^2\bigr],
\end{equation*}
and the analogous combinations for the leptonic $W$ and the two candidate
$W$-boson masses. Two further channels carry the coplanarity scalar
$\grade{M\reverse{M}}{0}$ of the meet $M$ and the pseudoscalar sign, ten
channels in all (App.~\ref{app:tokens}). By Theorem~\ref{lemma:cm-collapse}
every scalar channel of the pairing tokens except the pseudoscalar sign is a
function of $\{p_i\!\cdot\!p_j, m_i^2\}$ and not a new Lorentz invariant. The
pairing tokens add the grade-two and grade-three covariants of the
multi-resonance topology explicitly at the input layer, where the network
would otherwise build them internally.

\paragraph{Relation to prior work.}
At the per-particle level the token shape
$T_i = \grade{T_i}{0}\oplus\grade{T_i}{1}$ of \netGa\ coincides with the
multivector input encoding of L-GATr~\cite{Brehmer:2024yqw}, in which the
higher-grade covariants are built inside the network. The present
construction differs from the geometric-algebra transformer line of
work~\cite{Brehmer:2024yqw, Spinner:2024hjm, brehmer2023geometric} in three
respects. First, the event-level pairing tokens place grade-two and
grade-three covariants of the multi-resonance topology at the input layer,
whereas L-GATr builds such content inside the network. Second, the beam
enters as the single bivector $\gamma_{03}$, without a time vector. Reference~\cite{Spinner:2024hjm}
encodes the beam as the bivector of the plane orthogonal to the beam and
adds the time direction $(1,0,0,0)$ as a vector.
Reference~\cite{Brehmer:2024yqw} encodes it either as two vectors
$(0,0,0,\pm1)$ or as one bivector of the transverse plane, finds similar
performance for both choices, and also adds the time vector. Our
$\gamma_{03}$ is the Hodge dual of the transverse-plane bivector
$\gamma_{12}$, and Sec.~\ref{sec:ablation} tests bivector and vector
reference tokens against each other. Third, the read-out is resolved by grade
(Sec.~\ref{sec:outlook}). Fine-tuning in the L-GATr line re-initialises
the final output layer~\cite{Brehmer:2024yqw} without exposing per-grade
outputs, whereas the present interface routes each task to the grade that
carries the relevant covariant content.

\paragraph{Reporting.}
For the binary classification of $tW$ against $t\bar t$ we report ROC curves
with the area under the curve (AUC) quoted next to each label. Every number,
including those of \netRef, is computed on events held back from the
training of both networks, in their natural class proportions. The class imbalance is handled by a balanced sampler during
training and enters the AUC only through its variance. The
discriminant-output histograms, one per input representation, show three
parton-level samples passed through the trained network: the single-resonant
$tW$ sample (Class~A), the double-resonant $t\bar t$ sample (Class~B), and
the full gauge-invariant $tWb$ matrix-element sample of
Ref.~\cite{Boos:2023kpp}. The last curve shows how the network behaves on
the realistic event population, which combines both contributions and their
interference. The histograms are normalised to the matrix-element cross sections of the
three samples after the event selection, so the curves are
$\mathrm{d}\sigma/\mathrm{d}D$ and their relative heights reflect physical
weights. A fourth line gives the sum of the $tW$ and $t\bar t$ curves, and its
difference from the full matrix element is the interference of the two
classes.

\paragraph{Expected performance.}
By Theorem~\ref{lemma:cm-collapse}, the inner-product features of
\netRef\ already span the Lorentz-invariant scalar content available to a
parton-level classifier, apart from the pseudoscalar sign, so the two
networks are expected to reach comparable asymptotic AUC. A gain of
\netGa\ could come from two sources. The internal bivector and trivector
channels and the pairing tokens expose the assignment of the two $b$-jets
to the candidate top resonances algebraically, instead of through scalar
reconstruction shortcuts. The network also has covariant access to all
$p_i\wedge p_j$ and $p_i\wedge p_j\wedge p_k$ structures in a single
forward pass, an inductive bias that the scalar inputs of \netRef\ do not
provide.

\paragraph{AUC and ROC.}
Table~\ref{tab:tWb-auc} and Fig.~\ref{fig:tWb-roc} summarise the
results. Both networks use the same generated samples and are evaluated
at parton level. For \netGa\ we report the mean and the spread over five
training runs that differ only in the random seed, which sets the weight
initialisation, the train/validation split and the batch order together.
Figure~\ref{fig:tWb-roc} shows one of those runs, whose areas lie within
that spread of the means in Table~\ref{tab:tWb-auc}.

\paragraph{Discriminator distributions.}
In both input configurations of Fig.~\ref{fig:tWb-discr} the dominant
peak at $D\!\to\!1$ on the $tW$ curve has a kinematic origin. It is populated
almost entirely by events with $\sqrt{\hat s}\!<\!2m_t$, where a $t\bar t$
pair is energetically inaccessible and the logit saturates. The peak is
present for both inputs and is sharper with the reference bivector, which is
the effect recorded by the AUC in Table~\ref{tab:tWb-auc}. The percent-level
wrong-side tails are a property of the truth definition. A $tW$ event with
both top-line invariants accidentally on the Breit--Wigner pole
($m_t\!=\!172.5\pm 2$~GeV in both legs) is kinematically indistinguishable
from $t\bar t$. Conversely, $t\bar t$ events with one top line strongly off
shell ($m_t\!\gtrsim\!300$~GeV at the $95\%$ percentile of the
misclassified subsample) populate the same phase space as $tW$. The two
classes interfere as gauge-invariant pieces of the same matrix element, and
the diagram-removal scheme leaves a small irreducible residue at the
matrix-element level.

\begin{table}[t]
\centering
\small
\begin{tabular}{@{}l@{\ }r@{\ \ }c@{\ \ }c@{\ \ }c@{}}
\toprule
Input representation & Params & AUC & s.e. & $\Delta$AUC vs.\ four-momenta \\
\midrule
grade $0$ only & $151\,489$ & $0.9627\pm0.0003$ & $0.0005$ & $-0.00422\pm0.00025$$^{\ast}$ \\
$+$ pairing invariants & $151\,489$ & $0.9650\pm0.0003$ & $0.0005$ & $-0.00195\pm0.00025$$^{\ast}$ \\
$+$ pairing $1,2,3$ & $151\,489$ & $0.9649\pm0.0008$ & $0.0005$ & $-0.00206\pm0.00031$$^{\ast}$ \\
$+$ grade-$4$ join & $196\,609$ & $0.9653\pm0.0013$ & $0.0005$ & $-0.00160\pm0.00046$$^{\ast}$ \\
$+$ pairing $1,3$ (no meet) & $151\,489$ & $0.9657\pm0.0009$ & $0.0005$ & $-0.00122\pm0.00039$$^{\ast}$ \\
grade $0\oplus1$ four-momenta & $149\,409$ & $0.9669\pm0.0007$ & $0.0005$ & --- \\
$+$ ref.\ tokens $\gamma_0,\gamma_3$ & $149\,409$ & $0.9670\pm0.0004$ & $0.0005$ & $+0.00009\pm0.00021$ \\
\textbf{$+$ ref.\ bivector $\gamma_{03}$} & $149\,409$ & $0.9744\pm0.0007$ & $0.0004$ & $+0.00751\pm0.00055$$^{\ast}$ \\
\midrule
\netRef\ \cite{Boos:2023kpp}, 75 feat. & $539\,501$ & $0.9586$ & $0.0002$ & $-0.0083$ \\
\bottomrule
\end{tabular}
\caption{Parton-level classification of single-resonant $tW$ against
double-resonant $t\bar t$ topologies in $pp\to tWb$ at $\sqrt s=14$~TeV,
$p_T^{j_4}>10$~GeV. Rows differ only in the content of the input tensor, except the
grade-four join row, which adds one layer.
Each configuration is trained with five random seeds and scored on
held-back events with final-epoch weights. Column~5 is the difference to
the four-momentum run of the same seed, averaged over seeds, and an
asterisk marks $|\Delta|>2\sigma$. Rows with equal parameter counts are
directly comparable, since a fixed reference token carries no weights. The
grade-four join row has no capacity-matched partner. The \netRef\ row is a
single training without seed spread.}
\label{tab:tWb-auc}
\end{table}

\begin{figure}[t]
\centering
\includegraphics[width=\linewidth]{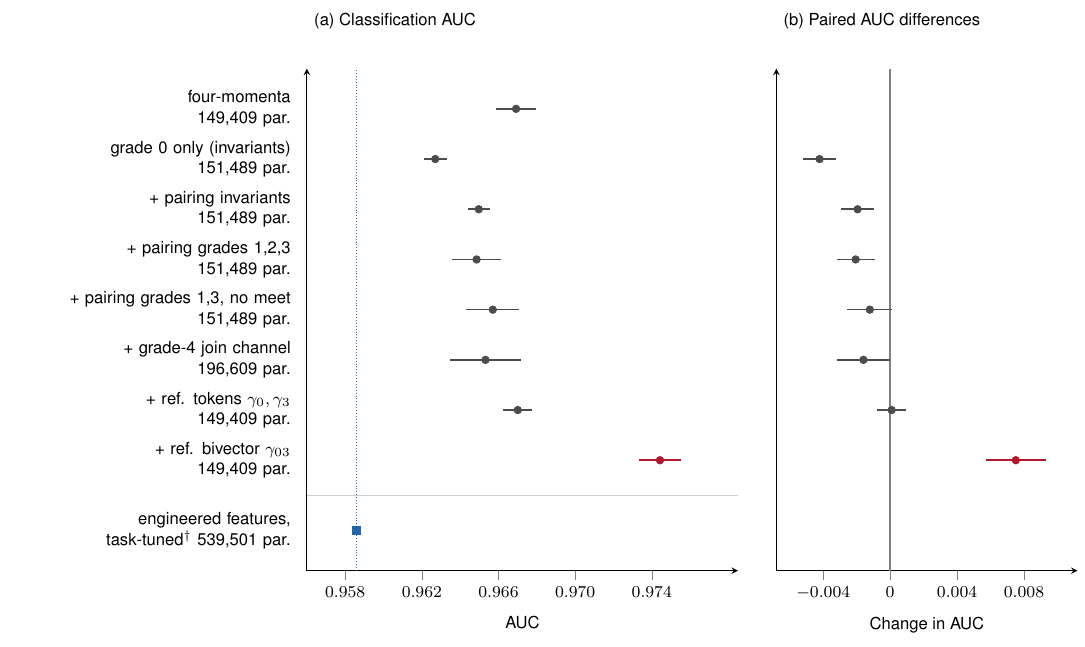}
\caption{\emph{(a)} AUC of each input representation, with the
engineered-feature baseline \netRef~\cite{Boos:2008sdz, Dudko:2020qas,
Boos:2023kpp} for comparison. The generic algebraic input exceeds it with
less than a third of the parameters (Sec.~\ref{sec:ablation}, Table~\ref{tab:tWb-auc}).
\emph{(b)} Change in AUC relative to the four-momentum baseline, formed seed
by seed and averaged, with architecture, data and split fixed. Trainable
parameters are given beside each row, and rows with equal counts are
directly comparable. Appendix~\ref{app:tokens} gives the two input layouts
cell by cell. Intervals are $95\%$ Student intervals on four degrees of
freedom.}
\label{fig:ablation}
\end{figure}

\subsection{Input ablation}
\label{sec:ablation}

The ablation addresses two questions. The first is whether a
representation designed to be uniform across processes is competitive
with features chosen for one process. The generic input of four-momenta
and object-type labels, with no feature engineering for $pp\to tWb$,
reaches an AUC of about $0.967$ with less than a third of the parameters
of \netRef. \netRef\ reaches $0.9586$ with an input recipe optimised for this
separation~\cite{Boos:2008sdz, Dudko:2020qas, Boos:2023kpp}. With
the fixed reference bivector the same network reaches about $0.974$
(Table~\ref{tab:tWb-auc}).

The second question is which algebraic content is useful once the
representation is fixed. Whether the grade-two and grade-three input
content motivated in Sec.~\ref{sec:intro} helps on \emph{this} final state
is tested on a ladder of eight input representations of the same network
that differ in nothing else. Each is trained with five random seeds
(Table~\ref{tab:tWb-auc}, Fig.~\ref{fig:ablation}). The ladder has two
branches, which differ in input layout. The upper rows share the wide
$22$-channel projection of the pairing layout. Its grade-$0$-only rung has
every covariant cell zeroed, and the `$+$' rungs restore the four-momenta
and vary the pairing content. The lower rows use the $9$-channel
projection of the four-momentum baseline. The $2\,080$-parameter
difference between the branches sits entirely in the width of that
projection, whereas an added token costs no parameters
(App.~\ref{app:tokens}).

The effects under test are of the order of the seed-to-seed scatter.
Each pair of configurations is therefore compared at a common seed, where
both see identical events, and every number quoted below is the average of
such paired differences.

\paragraph{Covariant channels of the pairing tokens.}
The comparison that isolates the grade content keeps the two pairing tokens, with their ten
grade-zero invariants and the same $151\,489$ parameters, and zeroes only
their covariant channels (grades one, two and three). Zeroing them changes
the AUC by $+0.000115\pm0.000251$, $t(4)=0.46$, $p=0.67$, with a $95\%$
interval of $[-0.00058,+0.00081]$, so five paired runs resolve no effect.

This result concerns \emph{where} the higher grades are formed. The
network is Lorentz-equivariant on multivector channels, so when the
grade-two and grade-three objects are not supplied it forms them from the
momenta it does receive. Attention mixes momenta across tokens, and the
geometric product wedges them within a token. On this final state the
internal construction is as good as the explicit one. Higher-grade content that the network cannot build from the momenta does
change the result, as the fixed reference bivector below shows.

The result is specific to this small example. The final state has six
particles and two candidate top decays, and the network can form the few
higher-grade objects that the separation needs. The pairing tokens instead
supply the candidate objects of both $b$-assignments in full, including parts
that this separation does not need. A foundation model has to cover many final
states with several different resonances, and the number of candidate
combinations grows combinatorially with the multiplicity
(Sec.~\ref{sec:practical}). Forming the relevant higher grades inside the
network may then be much harder, and computing them at the input could improve
the efficiency of the network noticeably.

Two coarser removals give larger shifts. Dropping the pairing tokens
outright gives $+0.00206\pm0.00031$ ($t(4)=6.70$, $p=0.003$). Dropping
only the grade-two meet while keeping the grade-three trivectors gives
$+0.00085\pm0.00032$ ($t(4)=2.62$, $p=0.06$, interval
$[-0.00005,+0.00174]$), a suggestive point estimate but not an
established effect. Neither removal isolates the grade content.
Removing the tokens also removes two sequence positions, narrows the
input projection from $22$ channels to $9$, and changes both the
attention graph and the mean pooling of the read-out.

Theorem~\ref{lemma:cm-collapse} accounts for the null result of the
isolated comparison, since the covariant pairing channels carry no Lorentz
invariant beyond functions of the momenta the network receives. The
theorem says nothing about optimisation. Channels that add no invariant can help or hurt a network of finite capacity trained for a finite time, so the small
shifts above neither follow from the theorem nor contradict it.

\paragraph{Discriminating grades in this final state.}
Which part of the multivector carries the discrimination depends on the
final state, and this benchmark was chosen as an extreme case, since the
class difference is carried mainly by one resonant three-particle
combination. Both $W$ bosons are on shell in \emph{both} classes, so the
invariant mass of their grade-two decay planes is distributed alike. As
single observables, $m(\mu^-\bar\nu_\mu)$ separates the classes with AUC
$0.5077$ and the distance to the $W$ pole with $0.5004$. The grade-two norm
of the same plane gives the same AUC as the mass, since for massless
partons it equals $(m^2/2)^2$. These one-dimensional tests do not exclude class
information in the orientation of the plane or in its correlation with the
rest of the event. The leading class difference is the three-particle
combination of the leptonic top candidate, resonant only in the $t\bar t$
class. In the algebra this decay is described by the grade-two planes of its
two-body subsystems and by its grade-three oriented volume, while its
invariant mass is the long-known grade-zero variable of the same
combination. The bare mass
separates the classes only at $0.5189$, because the resonance lies inside the
continuum with which it competes, and the class is not a monotonic function
of the mass. The pole distance $-|m-m_t|$ of the same invariant gives $0.7442$,
and the two top candidates together give $0.9570$. The grade-$k$ object
supplies the argument, and the network has to learn the non-monotonic
function of it. Grade four is populated, and its sign distribution in this
sample is symmetric within the sample precision (pseudoscalar AUC
$0.5093$). Section~\ref{sec:practical} explains why this checks the
contraction order and is not a test of CP.

In other processes the discriminating content sits in other grades. A
two-body resonance search ($H\to b\bar b$, $Z'\to\ell\ell$) puts it in the
grade-two plane of the pair, the sector that is inert here. A CP-violating
coupling, in $t\bar t$ production or in $H\to\tau\tau$, activates the
grade-four sign, which no combination of inner products reproduces and
which is the one channel Theorem~\ref{lemma:cm-collapse} leaves open.
Angular-correlation measurements of top-quark spin use the grade-two
bivector of the production plane. Conclusions about individual grades
therefore hold \emph{for this final state} only. What carries over to
other processes is the criterion that a grade-$k$ object reconstructible
from the visible momenta adds no invariant.

\paragraph{Fixed reference bivector.}
Adding one fixed token carrying the beam bivector $\gamma_{03}$ to the
grade-$0\oplus1$ configuration improves the AUC by about $0.0075$ at
identical parameter count (Table~\ref{tab:tWb-auc}). The two \emph{vector}
reference tokens $\gamma_0$ and $\gamma_3$ give no measurable change over
the same baseline, and the bivector-minus-vectors difference is
$+0.00742\pm0.00047$. As shown in Sec.~\ref{sec:symm}, $\gamma_{03}$
is invariant under both generators of $\Hlhc$ (Eq.~\eqref{eq:g03inv}),
whereas a longitudinal boost mixes $\gamma_0$
and $\gamma_3$. One fixed bivector, which no product of the
measured momenta generates and which preserves $\Hlhc$, thus improves the
classifier, while the two vectors, which break longitudinal-boost
invariance, do not. We read the L-GATr results in the same way. Both L-GATr
encodings add the time vector, which is not invariant under longitudinal
boosts, so either beam encoding reduces the symmetry to azimuthal
rotations. This reading is consistent with the similar performance of the two
encodings in Ref.~\cite{Brehmer:2024yqw}.

The two results differ in the origin of the grade content. A
grade-$k$ object built from the measured momenta adds no Lorentz-invariant
information, and on this benchmark its explicit presentation changed
nothing measurable. A grade-$k$ object encoding structure that those
momenta do not contain, here the beam axis, adds information on this
benchmark.

\paragraph{Reduced training statistics.} An inductive bias is expected to
help most where data are scarce, so the comparisons were repeated with
the training set thinned to one third and the evaluation events left
identical. The epoch and warmup budgets were scaled with the thinning so
that every fraction sees the same number of optimiser steps. At 3.1M
training events instead of 9.3M, and over three seeds, removing the
candidate-pairing content gains $+0.00215\pm0.00069$ and the reference
bivector gains $+0.00741\pm0.00058$. Both agree with their full-statistics
values ($0.12\sigma$ and $0.13\sigma$). The reference-bivector gain is
resolved at reduced statistics as well ($t(2)=12.9$, $p=0.006$). The pairing
removal is not ($t(2)=3.1$, $p=0.09$, interval $[-0.00084,+0.00513]$).
Neither effect grows visibly as data become scarce, but the
reduced-statistics intervals are wide enough that a moderate growth of the
pairing effect would have gone undetected.

In the grade-zero-only limit the network sees Lorentz invariants alone and
loses about $0.0042$ against grade-$0\oplus1$ four-momenta
(Table~\ref{tab:tWb-auc}), so the covariant grade-one input contributes
measurably. The grade-four join row adds to the pairing-$1,2,3$ configuration one
residual block built on the outer product, the join of
Eq.~\eqref{eq:meetjoin}, which forms grade-four objects from lower grades.
The attention dot products and softmax in this block are evaluated in
double precision.
It carries $32\%$ more parameters and still loses about $0.0016$ against
the same baseline, so on this final state the extra layer does not justify
its additional parameters.

\begin{figure}[t]
\centering
\includegraphics[width=0.72\linewidth]{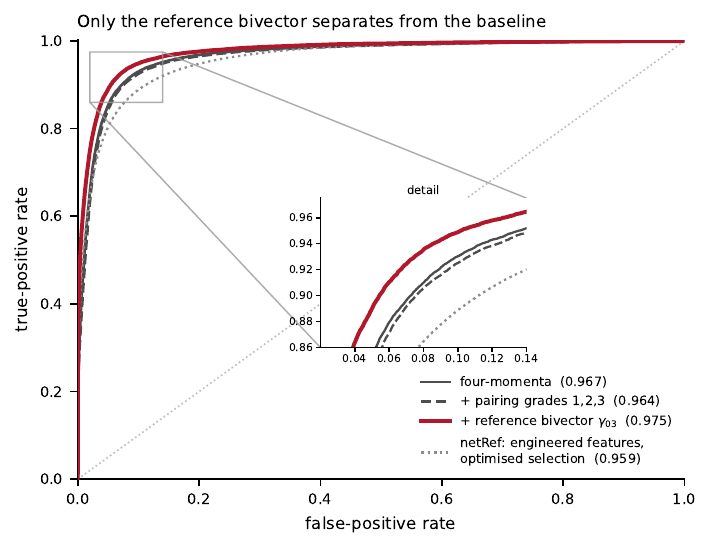}
\caption{ROC curves for the parton-level discrimination of
single-resonant $tW$ against double-resonant $t\bar t$ topologies of
$pp\to tWb$ with $p_T^{j_4}>10$~GeV, for three input representations of
the same network (one seed of five). The inset magnifies the region where
they separate. Candidate-pairing content built from the final-state
momenta leaves the curve on top of the four-momentum baseline, and one
fixed reference bivector $\gamma_{03}$ lifts it. The lower grey curve is
\netRef~\cite{Boos:2023kpp}, the engineered-feature baseline network on the
hand-optimised input recipes of Refs.~\cite{Boos:2008sdz, Dudko:2020qas}.
The dotted diagonal is the random classifier.}
\label{fig:tWb-roc}
\end{figure}

\begin{figure}[t]
\centering
\includegraphics[width=\linewidth]{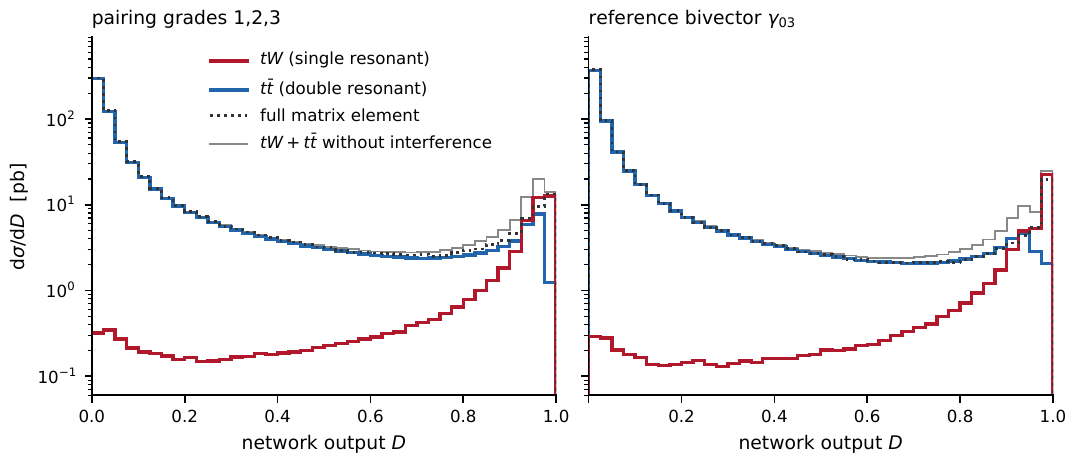}
\caption{Discriminator distributions, normalised to the parton-level cross sections
after the event selection, for two input representations of the same network, with
candidate-pairing content of grades one to three (left) and with the fixed
reference bivector $\gamma_{03}$ (right). Both clearly separate the
single-resonant $tW$ and double-resonant $t\bar t$ samples. The full matrix element (dotted) coincides with the $t\bar t$ curve at small
$D$ and runs close to the $tW$ curve at $D\to1$. The thin grey line is the sum
of the $tW$ and $t\bar t$ curves without their interference, which is negative
at large $D$. With the bivector, $58\%$ of $tW$ events
have $D>0.95$, against $52\%$ on the left, and $68\%$ of $t\bar t$ events
have $D<0.05$, against $62\%$.}
\label{fig:tWb-discr}
\end{figure}

\section{Foundation models for collider events}
\label{sec:outlook}

Foundation models for collider physics, pre-trained once and fine-tuned for
many tasks, are now being explored~\cite{Mikuni:2024qsr}. The multivector
$\evMV$ is meant as a uniform input layer for them.

\paragraph{Properties as an input layer.} Three properties make $\evMV$
suitable for this role. First, it does not depend on the physics process.
The multivector of Eq.~\eqref{eq:evMV} encodes $W^+W^-$, $t\bar t$, $H\to
b\bar b$, EFT signals and BSM search regions alike, and only its
decomposition on $\Vflav$ changes. New states add basis elements to
$\Vflav$, while modified interactions of Standard Model particles (the SMEFT
case) change the kinematics and hence the grade content of $\evMV$ at fixed
$\Vflav$.
Second, pre-training on $\evMV$ keeps the grade structure end to end, so a
task that needs strict invariance reads the grade-zero coefficients and a
task that needs full covariance reads the whole multivector. With lab-frame
recipes~\cite{Boos:2008sdz, Dudko:2020qas} the choice of frame propagates to
every downstream task, and with raw four-momenta every task has to relearn
the representation of $\Spin^+(1,3)$.
Third, the symmetry restrictions of a hadron collider can be built into
$\evMV$. The beam axis reduces the symmetry to $\Hlhc$
(Eq.~\eqref{eq:Hlhc}), missing transverse momentum is compatible only with
this subgroup (Sec.~\ref{sec:flav}), and both restrictions correspond to a
sub-decomposition of $\evMV$ along the $2{+}2$ split. Of three possible
mechanisms, the fixed $\gamma_{03}$ reference token ($\Hlhc$-invariant by
Eq.~\eqref{eq:g03inv}) is implemented and measured
(Sec.~\ref{sec:ablation}). Layers adapted to the split of
Eq.~\eqref{eq:22split} and the masked grade-one pseudo-object of
Eq.~\eqref{eq:MET} for the missing momentum are specified but not yet
tested. Block-diagonal
weights on the $2{+}2$ split are not enough for such layers. An
$\Hlhc$-equivariant layer needs each grade decomposed into
$\Hlhc$-irreducible components, with the weights as intertwiners between
them, and we have not worked this out for all five grades.

\paragraph{Relation to existing work.} Existing proposals for collider
foundation models (masked-particle pre-training~\cite{Golling:2024abg},
contrastive jet self-supervision~\cite{Dillon:2021gag}, multi-task event
models~\cite{Mikuni:2024qsr}) impose no architectural equivariance. They
operate on raw four-momenta plus flat scalar channels for particle type,
charge, $b$-tag and process label. These channels are the content of
$\Vflav$, and in $\Cl(1,3)\otimes\Vflav$ they become one factor of a single
structure, with the same token interface in pre-training and in every
downstream task. The backbone assumed here is L-GATr, the
Lorentz-equivariant geometric-algebra transformer on the same algebra.
Ref.~\cite{Spinner:2024hjm} trains it for single tasks and names
pre-training across processes as future work, and
Ref.~\cite{Brehmer:2024yqw} pre-trains it on JetClass and fine-tunes it for
top tagging.

\paragraph{Grade-resolved pre-training and read-out.} The grade
decomposition supplies one masking target and one read-out head per grade.
Pre-training predicts masked components of $\evMV$ from the unmasked ones.
Grade zero holds the Lorentz invariants, grade one the mass-shell
constraints and matrix-element correlations, grade two the decay-plane
geometry, grade three the $T$-odd oriented volume~\cite{Atwood:2000tu,
Bernreuther:2015yna}, and grade four the CP structure. The loss curve of
each target therefore shows which content the encoder has learned.
Fine-tuning attaches a single head to the unchanged encoder. Classification
and scalar regression read the grade-zero coefficients of the last hidden
state, and four-momentum regression reads the grade-one component. A
CP-asymmetry test compares the sign of the grade-four coefficient with a
CP-even null.

\subsection{Practical obstacles}
\label{sec:practical}

A graded representation meets four practical obstacles. We quantify three of
them on our own runs.

\paragraph{Per-grade dynamic range.} A wedge product of $k$ momenta has mass
dimension $k$, so the algebra fixes the numerical scale of each input
channel. In the pairing tokens of Sec.~\ref{sec:tWb} the raw entries span
over $25$ orders of magnitude in one input tensor. They range from $\mathcal
O(10)$~GeV$^2$ at the low tail of a grade-two invariant to
${\sim}7\times10^{25}$~GeV$^{12}$ at the high tail of the coplanarity
scalar. The coplanarity scalar is the squared norm $\grade{M\reverse{M}}{0}$
(Eq.~\eqref{eq:bivinner}) of the meet $M = T_{\rm had}\cap T_{\rm lep}$,
with six powers of momentum in each factor. The spread is dimensional: the
median of each object taken to the power $1/k$ lies between $43$ and
$190$~GeV at every mass dimension from one to twelve. Per-grade
normalisation removes the spread between objects, but up to $6$ orders of
magnitude remain within one object, and a network must either learn across
that range or be given a grade-aware scaling. A foundation model shared
across processes cannot fix these scales in advance, because they follow the
mass and $\hat s$ scales of each final state.

\paragraph{Pseudoscalar sign and normalisation.} By
Theorem~\ref{lemma:cm-collapse}, the \emph{sign} of the grade-four
pseudoscalar is the only Lorentz-invariant channel not reducible to inner
products. A normalisation that acts on magnitudes (per-grade rescaling,
layer normalisation, a $\log$ compression of the range above) can discard
it, and a lost sign cannot be recovered. Any transformation of the
grade-four channel must preserve the sign.

\paragraph{Contraction order and apparent CP asymmetry.} Our sign test
checks the contraction order, not CP. The sample is the fixed $\mu^-$
channel, and CP maps it onto the charge-conjugate $\mu^+$ channel, which we
do not generate, so the sign distribution here neither confirms nor
contradicts CP conservation. It does show whether the ordering convention
biases the sign. With the two same-parent light partons ordered
by energy, the positive fraction on the $268$k-event $tW$ sample is
$0.4989$, $1.2\sigma$ from one half. Contracting instead the PDG-labelled
slots $\epsilon(p_\mu, p_{\bar\nu}, p_u, p_{\bar d})$ gives $0.5115$,
$11.9\sigma$ from one half. Exchanging the two partons of the same parent
flips the sign of the contraction exactly. A labelling correlated with the
kinematics therefore produces an asymmetry that a CP-violation search would
report as a signal. An analysis using this channel must fix the contraction
order by a physical rule, not by truth flavour, and test CP on a
charge-conjugate pair of samples.

\paragraph{Computational cost.} The number of multivector channels, and with
it the cost of attention, grows combinatorially with the final-state
multiplicity. The $\binom{N}{2}$ bivectors, $\binom{N}{3}$ trivectors and
$\binom{N}{4}$ pseudoscalars number $50$ for the $N=6$ partons of the
present $tWb$ final state and $154$ for a $2\to8$ process. For the $N=12$
partons of four-top production or $t\bar ttWb$ they number $781$, fifteen
times the present count. In our runs the two pairing tokens raise the epoch
time of the four-momentum baseline from $259$ to $361$~s, by $39\%$, and
also lower the AUC (Sec.~\ref{sec:ablation}). The reference bivector, which
helps, is one fixed token independent of $N$ and costs no parameters and no
measurable time. A single-task network can keep only the sectors its task
reads, whereas a foundation model spanning processes must budget for the
full decomposition and leave the pruning to fine-tuning.

\section{Limitations and future directions}
\label{sec:limits}

\paragraph{Scope of the construction.} The invariant content of the higher
grades is bounded by a property of the algebra, the Cayley--Menger collapse
(Theorem~\ref{lemma:cm-collapse} and the grade-four paragraph of
Sec.~\ref{sec:grades}). Three kinds of information lie outside the
construction or enter it in reduced form. Lab-frame event shapes such as
sphericity, thrust and aplanarity lie in the orthogonal complement of the
bivector channel (Sec.~\ref{sec:grades}). The $b$-tag information is reduced
to the binary slot $|bj\rangle\in\mathbb F$ (Sec.~\ref{sec:flav}). Fermion
spin enters only implicitly, through the bivector grade and the
second-moment structure of production and decay~\cite{Mahlon:1999gz,
Bernreuther:2013aga, Bernreuther:2015yna}, with no spinor-valued fields.

\paragraph{Conformal $\Cl(2,4)$ for multi-resonance event topologies.} An
algebra beyond $\Cl(1,3)$ is physically motivated once an event has more
than one heavy resonance, as in a $t\bar t$ event, four-top production or a
ladder of $W$-mediated cascades. The kinematic support of such an event is
described by the \emph{intersections} of mass shells $(p_a+p_b)^2 = M^2$. In
$\Cl(1,3)$ these reduce to scalar Cayley--Menger combinations of inner
products (Theorem~\ref{lemma:cm-collapse}). A configuration with one
resonance on shell and the other off shell then has no geometric object of
its own, although it occupies a distinct sub-manifold of phase space. The
conformal embedding into $\Cl(2,4)$~\cite{Doran:2003} keeps this content
algebraic. Two auxiliary null vectors $n_\infty, n_o$ extend the generating
space to six dimensions. Each four-momentum is mapped to the null conformal
vector
\begin{equation}
 P \;=\; n_o + p + \tfrac{1}{2}\,p^2\,n_\infty,
 \qquad
 P_i\!\cdot\!P_j \;=\; -\tfrac{1}{2}\,(p_i-p_j)^2,
 \label{eq:conformal-embed}
\end{equation}
so that the inner product of two conformal images is proportional to the
squared difference $(p_i-p_j)^2$ of the four-momenta~\cite[\S 10.2]{Doran:2003}. A mass shell
$(p_W+p_b)^2=m_t^2$ becomes one blade of fixed grade~\cite{Doran:2003,
Hestenes:1984}, and the meet of two shells, $\Sigma_{t_1}\!\cap\Sigma_{t_2}
= (\Sigma_{t_1}^\star \wedge\Sigma_{t_2}^\star)^{-\star}$, is a lower-grade
blade that carries the intersection of the two top hyperboloids. The sign of
its squared norm separates events with two real on-shell top assignments
from events at threshold and events with no real intersection. The
$b$-assignment of Sec.~\ref{sec:tWb} becomes one incidence test
$(W+p_{b_a})\!\cdot\!\Sigma_t \approx 0$ between the conformal image of the
$(W+b_a)$ candidate and the top mass-shell blade $\Sigma_t$. The
construction and the incidence test keep their form for any chain of $N$
nested resonances, as the multi-resonance counterpart of the decay-plane
bivector of Sec.~\ref{sec:grades}.

The embedding does \emph{not} enlarge the ring of Lorentz-invariant scalars.
By the First Fundamental Theorem of invariant theory for the orthogonal
group~\cite{Weyl:1939}, every Lorentz invariant of the $\{P_i\}$ and the
null pair $(n_o, n_\infty)$ reduces to the orientation sign of
Theorem~\ref{lemma:cm-collapse} and a polynomial in their inner products.
These are $\{P_i\!\cdot\!P_j\} = \{-\tfrac{1}{2}(p_i-p_j)^2\}$ and
$\{P_i\!\cdot\!n_o\} = \{-\tfrac{1}{2}m_i^2\}$, with $P_i\!\cdot\!n_\infty =
-1$, and are themselves polynomial in $\{p_i\!\cdot\!p_j, m_i^2\}$. The
extension is therefore representational. It exposes the multi-resonance
geometry as algebraic objects that an equivariant network can take as input,
as it takes the pairing tokens of Sec.~\ref{sec:tWb}. Those tokens carry the
corresponding $\Cl(1,3)$ content and gave no gain in the small $tWb$ example (Sec.~\ref{sec:ablation}), where
the network forms this content in its layers. The conformal objects may still help where the
incidence test cannot be reconstructed from the visible momenta, in
topologies with more than one unmeasured momentum. Such topologies are the
next case to study.

\paragraph{Other algebraic extensions.} The multi-particle Clifford bundle
$\Cl(1,3)^{\otimes N}$ (Variant~C of Sec.~\ref{sec:flav}) makes the
per-token factorisation explicit, at the price of an algebra that depends on
$N$. Further open directions are the self-duality decomposition of
bivectors, the Weyl-spinor parallel via
$\Spin^+(1,3)\simeq\SLgroup(2,\mathbb C)$, the hyperbolic geometry $\mathbb
H^1\times\mathbb S^1$ of the rapidity-azimuth plane, and the projective
extension $\Cl(1,3,1)$ for an explicit vertex-level encoding.

\paragraph{Demonstration caveats.} The $pp\to tWb$ study of
Sec.~\ref{sec:tWb} demonstrates the input representation on a parton-level
leading-order matrix element, without parton shower or detector simulation,
and is neither a phenomenological estimate nor an implementation of a
foundation model. The CP-odd channel of
Sec.~\ref{sec:grades} has zero expectation in CP-symmetric tree-level
configurations, and testing it at the network level would need CP-asymmetric
matrix-element samples. Cross-process generalisation, larger and more
diverse training samples, detector simulation and a systematic comparison
with Lorentz-equivariant baselines~\cite{Brehmer:2024yqw, Spinner:2024hjm,
brehmer2023geometric} need dedicated studies.

\paragraph{HiGEN foundation-model programme.} The present construction is
the input-layer specification of the HiGEN programme of
Sec.~\ref{sec:outlook}. Its encoder, per-grade pre-training strategy,
downstream tasks (classification, regression, generative reconstruction) and
transfer protocol are left to forthcoming work.

\section{Summary}
\label{sec:summary}

We have described a hard hadron-collider event as a single geometric object,
the event multivector $\evMV$ (Eq.~\eqref{eq:evMV}) built from the
multivectors of its reconstructed objects. It gives a common framework for
three previously separate approaches to the network
input~\cite{Boos:2008sdz, Dudko:2020qas, Brehmer:2024yqw}. We give the
physical meaning of every grade and operation, a per-grade dictionary of
$32$ classical observables (App.~\ref{app:tables-c1}) and an inventory of
the symmetries acting on $\evMV$ (App.~\ref{app:tables-c2}). By
Theorem~\ref{lemma:cm-collapse}, the higher grades give no Lorentz-invariant
scalar beyond $\{p_i\!\cdot\!p_j,\,m_i^2\}$ except the CP-odd sign of the
pseudoscalar~\cite{Atwood:2000tu}. Grades two and three describe the decay
planes and decay volumes of the event compactly, as covariant objects. The
resonance-topology separation of $pp\to tWb$ (Sec.~\ref{sec:tWb}, limitations
in Sec.~\ref{sec:limits}) is a simple example of the use of the
representation. There the network forms the higher-grade objects it needs in its layers, and
pairing tokens built from the measured momenta give no further gain. A fixed
reference bivector $\gamma_{03}$ for the beam plane improves the classifier
(Sec.~\ref{sec:ablation}). Higher grades computed at the input could still
improve networks that cover many final states with several different
resonances. The multivector $\evMV$ is intended as
the input layer of a foundation-model programme (Sec.~\ref{sec:outlook}).
Its practical obstacles (Sec.~\ref{sec:practical}) include a dynamic range
of $25$ orders of magnitude, the loss of the grade-four sign under magnitude
normalisation, a false CP asymmetry from a label-ordered contraction, and
the combinatorial cost of the higher grades.

\data{The training code of the geometric-algebra network, the final version of this paper and instructions for obtaining the Monte-Carlo samples are available in the public companion repository~\cite{higen-repo}. The
repository also collects reference implementations of the algebraic input
layer and of the networks built on it for the foundation-model programme.}

\appendix
\numberwithin{equation}{section}

\theoremstyle{plain}

\section{Notation card}
\label{app:notation}

\begin{itemize}
 \item \textbf{Metric signature.} $\eta = \mathrm{diag}(+,-,-,-)$.
 \item \textbf{Basis of $\Cl(1,3)$.} Orthonormal grade-one basis
 $\{\gamma_0, \gamma_1, \gamma_2, \gamma_3\}$ with
 $\gamma_\mu \gamma_\nu + \gamma_\nu \gamma_\mu = 2\eta_{\mu\nu}$.
 \item \textbf{Pseudoscalar.} $I \equiv \gamma_0\gamma_1\gamma_2\gamma_3$,
 $I^2 = -1$.
 \item \textbf{Reverse.} $\reverse{X}$ reverses the order of all
 geometric products inside $X$. On an $r$-blade it is multiplication
 by $(-1)^{r(r-1)/2}$.
 \item \textbf{Dual.} $X^\star \equiv X\, I^{-1}$, with inverse
 $X^{-\star} \equiv X\, I$, so that $(X^\star)^{-\star} = X$. Both enter
 Eq.~\eqref{eq:meetjoin}. Since $I^{-1} = -I$, this dual differs by an
 overall sign from $X^\star = X\, I$ of Ref.~\cite{Doran:2003}.
 \item \textbf{Grade projection.} $\grade{X}{k}$ is the
 grade-$k$ component of $X$.
 \item \textbf{Geometric product.} $XY$. For grade-one elements it
 splits into the inner and the outer product,
 \begin{equation}
 p\,q \;=\; (p\!\cdot\!q) \,+\, (p\wedge q),
 \qquad
 p\!\cdot\!q \,=\, \tfrac12(pq + qp),
 \qquad
 p\wedge q \,=\, \tfrac12(pq - qp).
 \label{eq:gp-split}
 \end{equation}
 In components, with $p = p^\mu\gamma_\mu$, $q = q^\nu\gamma_\nu$ and
 $p\,q = p^\mu q^\nu\,\gamma_\mu\gamma_\nu$, symmetrising and
 antisymmetrising in $(\mu,\nu)$ gives
 \begin{equation}
 p\,q
 \;=\; \eta_{\mu\nu}\, p^\mu q^\nu
 \;+\; \sum_{\mu<\nu}\bigl(p^\mu q^\nu - p^\nu q^\mu\bigr)
 \,\gamma_\mu\gamma_\nu,
 \label{eq:gp-components}
 \end{equation}
 where the scalar part is the Minkowski inner product
 $p\!\cdot\!q = p^0 q^0 - \vec p\!\cdot\!\vec q$ and the bivector part
 has one component for each of the six oriented planes of
 $\grade{\Cl(1,3)}{2}$, three boosts and three rotations
 (Sec.~\ref{sec:grades}).
 \item \textbf{Blades.} An $r$-blade is the outer product of
 $r$ linearly independent grade-one elements. In $\Cl(1,3)$
 non-simple multivectors of grade $r$ exist only for $r=2$.
 Grade-three and grade-four elements are always blades, as Hodge
 duals of grades one and zero.
 \item \textbf{Lorentz group and rotor sandwich.} Proper orthochronous
 Lorentz group $\SOgroup^+(1,3)$, with double cover
 $\Spin^+(1,3) \simeq \SLgroup(2,\mathbb C)$. A Lorentz transformation
 acts on \emph{any} multivector by the same two-sided
 product
 \[
   X \;\longmapsto\; R\,X\,\reverse{R},
   \qquad R=\exp(B/2),\quad B\in\grade{\Cl(1,3)}{2},
   \qquad R\reverse R = 1 ,
 \]
 which is used to change frame, as for $\cos\theta^*$ in
 Sec.~\ref{sec:grades}, and to test invariance. An input is invariant
 under a subgroup if and only if it commutes with the corresponding
 rotors, as $\gamma_{03}$ does in Eq.~\eqref{eq:g03inv}. Rotors compose
 by multiplication, $R_2R_1$, the inverse is $\reverse R$, and the
 sandwich preserves the grade of $X$ because it is an algebra
 automorphism (App.~\ref{app:derivations}).
 \item \textbf{Bivector sign rules.} For $B \in \grade{\Cl(1,3)}{2}$
 the sign $\mathrm{sgn}(B^2)$ separates rotation generators
 ($B^2 < 0$ in our signature, e.g.\ $(\gamma_1\gamma_2)^2 = -1$)
 from boost generators ($B^2 > 0$, e.g.\
 $(\gamma_0\gamma_3)^2 = +1$). A generic bivector is the sum of
 one rotation and one boost in two mutually orthogonal commuting
 planes~\cite{Doran:2003}.
 \item \textbf{LHC residual subgroup.}
 $\Hlhc \equiv \SOgroup^+(1,1)_z \times \SOgroup(2)_\phi \subset
 \Spin^+(1,3)$, embedded through the rotor double cover
 (Eq.~\eqref{eq:Hlhc}).
 \item \textbf{Event multivector.} $\evMV \in
\bigl[\Cl(1,3)\otimes\Vflav\bigr]^{\oplus N}$ defined in
Eq.~\eqref{eq:evMV}, with one slot per reconstructed object.
 \item \textbf{Object-type space.} $\Vflav =
 \mathrm{span}_{\mathbb R}\{\,|f\rangle: f \in \mathbb F\,\}$
 with $\mathbb F$ given by Eq.~\eqref{eq:Vflav}.
 \item \textbf{Missing transverse momentum.} $\MET$,
 encoded as a pseudo-object by Eq.~\eqref{eq:MET}.
\end{itemize}

\section{Algebraic derivations}
\label{app:derivations}

\paragraph{Sign of $(p\wedge q)\,\reverse{(p\wedge q)}$ in
$(+,-,-,-)$.}
For grade-one $p, q \in \Cl(1,3)$ the bivector $B \equiv p\wedge q$
satisfies $\reverse{B} = -B$~\cite{Doran:2003}, since the reverse acts
as $(-1)^{r(r-1)/2}$ on a grade-$r$ element. Direct expansion in
$(+,-,-,-)$ gives
\begin{equation}
 (p\wedge q)^2
 \;=\; (p\!\cdot\!q)^2 \,-\, p^2\, q^2,
 \qquad
 (p\wedge q)\,\reverse{(p\wedge q)}
 \;=\; -\,(p\wedge q)^2
 \;=\; p^2\, q^2 \,-\, (p\!\cdot\!q)^2,
 \label{eq:gram2}
\end{equation}
so $B\reverse{B} \le 0$ for two future-directed time-like $p,q$ by the
reverse Cauchy--Schwarz inequality $(p\!\cdot\!q)^2 \ge p^2\, q^2$.
The non-negative Lorentz scalar of the two-blade is therefore
\begin{equation}
 |p\wedge q|^2
 \;\equiv\; -\,(p\wedge q)\,\reverse{(p\wedge q)}
 \;=\; (p\!\cdot\!q)^2 \,-\, p^2\, q^2 \;\ge\; 0.
 \label{eq:bivnorm}
\end{equation}
Equation~\eqref{eq:gram2} and the non-negativity of $|B|^2$ for two
time-like momenta are checked directly in the parameterisation
$p = \gamma_0$,
$q = \cosh\alpha\, \gamma_0 + \sinh\alpha\, \gamma_1$, which gives
\begin{equation*}
 p^2 q^2 - (p\!\cdot\!q)^2
 \;=\; 1 - \cosh^2\alpha
 \;=\; -\sinh^2\alpha,
 \qquad
 |B|^2 \;=\; \sinh^2\alpha \;\ge\; 0.
\end{equation*}

\paragraph{Gram-determinant identity for an $r$-blade.}
\begingroup\loosepar
The reversion contraction of an $r$-blade
$P_{i_1\cdots i_r} \equiv p_{i_1} \wedge \cdots \wedge p_{i_r}$
built from $r$ four-momenta is the $r\times r$ Gram determinant of
their inner products,
\begin{equation}
 P_{i_1\cdots i_r}\,\reverse{P}_{i_1\cdots i_r}
 \;=\; \det\!\bigl[\, p_{i_a}\!\cdot\!p_{i_b} \,\bigr]_{a,b=1}^{r},
 \label{eq:gramR}
\end{equation}
which for $r = 2$ is Eq.~\eqref{eq:gram2}, equal to $-|p\wedge q|^2$ of
Eq.~\eqref{eq:bivnorm}. The right-hand side carries its own sign through
the Minkowski inner products, so the contraction is not a squared norm.
For $r=4$ it equals $-\lambda^2$, with $\lambda$ the coefficient of
Eq.~\eqref{eq:pseudoscalar}. To derive Eq.~\eqref{eq:gramR}, choose an
orthonormal frame $\{e_a\}_{a=1}^{r}$ of the subspace spanned by
$\{p_{i_a}\}$, with $e_a \cdot e_b = \varsigma_a\, \delta_{ab}$ and
$\varsigma_a = \pm 1$ fixed by the signature of the subspace. Write
$p_{i_a} = M^a{}_b\, e_b$ with a non-singular real matrix $M$. Then
$P_{i_1\cdots i_r} = (\det M)\, e_1 \wedge \cdots \wedge e_r$, and
\begin{equation*}
 P_{i_1\cdots i_r}\,\reverse{P}_{i_1\cdots i_r}
 \;=\; (\det M)^2\,
 (e_1\wedge\cdots\wedge e_r)(e_r\wedge\cdots\wedge e_1)
 \;=\; (\det M)^2\,\prod_{a=1}^{r}\varsigma_a.
\end{equation*}
The Gram matrix of $\{p_{i_a}\}$ is
$[p_{i_a}\!\cdot\!p_{i_b}] = M\,\mathrm{diag}(\varsigma)\,M^\top$,
and its determinant $(\det M)^2 \prod_a \varsigma_a$ is the same
expression, with no extra signature factor.\par\endgroup

The proof of Theorem~\ref{lemma:cm-collapse} (Sec.~\ref{sec:grades})
uses these identities.

\begin{proof}
Single-grade contractions $\grade{X_r \reverse{X_r}}{0}$ reduce
directly to inner products. For $r=2$, Eq.~\eqref{eq:gram2} gives
\begin{equation*}
 (p\wedge q)\,\reverse{(p\wedge q)}
 \;=\; p^2\, q^2 - (p\!\cdot\!q)^2,
\end{equation*}
and for $r=3,4$ Eq.~\eqref{eq:gramR} gives, for an $r$-blade
$P=p_1\wedge\dots\wedge p_r$,
\begin{equation*}
 P\,\reverse P
 \;=\; \det\!\bigl[p_a\!\cdot\!p_b\bigr]_{a,b=1,\dots,r},
\end{equation*}
\begingroup\loosepar
a polynomial in the inner products. A general grade-$r$ multivector is
a sum of blades, so by linearity the same holds for any Lorentz scalar
built from it. The pseudoscalar $\grade{P_4}{4}= \lambda I$ enters
single-grade contractions only through $\lambda^2$. By the Gram
identity $\lambda^2 = -\det[p_a\!\cdot\!p_b]_{a,b=1,\dots,4}$, which
also reduces to inner products.\par\endgroup

For mixed-grade contractions $\grade{X_{r_1}\cdots X_{r_n}}{0}$ we use
the grade expansion of the product of two homogeneous
multivectors~\cite[eq.~(4.42)]{Doran:2003},
\begin{equation*}
 X_r\, X_s \;=\;
 \sum_{\substack{k=|r-s|\\ k\equiv r+s\,(\mathrm{mod}\,2)}}^{\,r+s}
 \grade{X_r X_s}{k}.
\end{equation*}
For two factors the grades on the right run in steps of two from
$|r_1-r_2|$ to $r_1+r_2$ and reach zero only when $r_1=r_2$, so
$\grade{X_{r_1} X_{r_2}}{0}$ vanishes whenever $r_1\ne r_2$. In the
case $r_1=r_2=r$ the projection $\grade{X_r Y_r}{0}$ differs by the
reversion sign from the bilinear pairing
$(X_r,Y_r)\mapsto \grade{X_r\reverse{Y_r}}{0}$ on grade-$r$
multivectors~\cite[\S4.1.3]{Doran:2003},
$\grade{X_r Y_r}{0}=(-1)^{r(r-1)/2}\grade{X_r\reverse{Y_r}}{0}$. The
polarisation
identity
\begin{equation*}
 2\grade{A\reverse B}{0}
 \;=\;
 \grade{(A+B)\reverse{(A+B)}}{0}
 - \grade{A\reverse A}{0} - \grade{B\reverse B}{0},
\end{equation*}
with $A,B$ both of grade $r$, writes
$\grade{X_{r_1}\reverse{X_{r_2}}}{0}$ as a combination of three
self-contractions $\grade{Z\reverse Z}{0}$ of grade $r$.
Longer chains are reduced by induction on $n$. The hypothesis is that
every \emph{grade-$r$ projection} $\grade{X_{r_2}\cdots X_{r_n}}{r}$
of a shorter chain is a combination of $r$-blades of the $p_i$
\emph{and their $I$-duals}. Its coefficients are polynomial in the
inner products and at most linear in the $\lambda_T$ (for $n=1$ it is
a blade with unit coefficient). The duals are needed because a product
can leave the blade class through the pseudoscalar. For instance,
$\grade{(p_1{\wedge}p_2{\wedge}p_3{\wedge}p_4)(p_5{\wedge}p_6{\wedge}p_7)}{1}
=\lambda_{1234}\,\bigl(I\,(p_5{\wedge}p_6{\wedge}p_7)\bigr)$ is a dual
vector, not an outer product of the momenta. The class is closed under
products: $I$ commutes with even and anticommutes with odd elements,
so after normal ordering at most one explicit $I$ remains, since
$I^2=-1$. A leftover $I$ in a grade-zero projection gives
$\grade{I\,Y_4}{0}=-\lambda_Y$, which is linear in $\lambda$. For the
inductive step write $\grade{X_{r_1}\cdots X_{r_n}}{0}=
\grade{X_{r_1}(X_{r_2}\cdots X_{r_n})}{0}$. By the two-factor argument
only the grade $t=r_1$ of the right factor contributes, so
$\grade{X_{r_1}\cdots X_{r_n}}{0}=
\grade{X_{r_1}\,\grade{X_{r_2}\cdots X_{r_n}}{r_1}}{0}$ is
the pairing of two grade-$r_1$ objects. The hypothesis, applied to the
grade-$r_1$ part $\grade{X_{r_2}\cdots X_{r_n}}{r_1}$ of a product of
length $n-1$, reduces it to self-contractions $\grade{Z\reverse Z}{0}$.
The grade-zero projection of any product of homogeneous-grade
multivectors is therefore a polynomial in the contractions
$\grade{Y_t \reverse{Y_t}}{0}$, which reduce to inner products, and in
the pseudoscalar coefficients $\lambda_T$. These come from a
grade-four factor $X_a$ or from the grade-four part of a product of
lower-grade factors, as in $\grade{(p_1\wedge p_2)(p_3\wedge p_4)}{4}$.
What counts is the number of $\lambda$ factors after the reduction, not
the number of grade-four $X_a$.
Even powers and pairwise products of the pseudoscalar coefficients
collapse by the \emph{mixed} Gram identity
\begin{equation}
 \grade{(p_1\wedge\dots\wedge p_r)\,
        \reverse{(q_1\wedge\dots\wedge q_r)}}{0}
 \;=\; \det\bigl(p_a\!\cdot\!q_b\bigr)_{a,b=1,\dots,r},
 \label{eq:gramMixed}
\end{equation}
of which Eq.~\eqref{eq:gramR} is the diagonal case $q_a=p_a$. Both
sides are multilinear and antisymmetric in each set of arguments, and
on basis blades both reduce by Eq.~\eqref{eq:cliff-anticomm} to the
same determinant of metric factors. For two pseudoscalars Eq.~\eqref{eq:gramMixed}
gives $\lambda_T\lambda_{T'} = -\det(p_a\!\cdot\!p_b)_{a\in T, b\in T'}$,
so the result is at most linear in the $\lambda$'s jointly and has
the form $P+\sum_T \lambda_T Q^T$ claimed in the theorem. Under spatial
inversion all $\lambda_T$ change sign together while every inner
product is unchanged, so the common orientation sign is the only
non-polynomial content (Sec.~\ref{sec:grades}).
\end{proof}

\paragraph{Action of $\Spin^+(1,3)$ on $\Cl(1,3)$.}
\begingroup\loosepar
A rotor $R=\exp(B/2)$ with bivector generator
$B\in\grade{\Cl(1,3)}{2}$ lies in the even subalgebra
$\Cl^+(1,3)$. Since $\reverse{B}=-B$ on any bivector,
\begin{equation*}
 \reverse{R}
 \;=\; \exp(\reverse{B}/2)
 \;=\; \exp(-B/2)
 \;=\; R^{-1},
\end{equation*}
and $R\,\reverse{R}=\reverse{R}\,R=1$. The rotor sandwich
\begin{equation}
 \mathrm{Ad}_R:\ X\mapsto R\,X\,\reverse{R}
 \label{eq:adR}
\end{equation}
is therefore an \emph{inner algebra automorphism} of $\Cl(1,3)$.
The map of Eq.~\eqref{eq:adR} is $\mathbb R$-linear in $X$. Inserting $\reverse{R}R=1$ between
adjacent factors shows that it respects the geometric product,
\begin{equation}
\begin{aligned}
 \mathrm{Ad}_R(X\,Y)
 \;&=\; R\,X\,Y\,\reverse{R}
 \;=\; R\,X\,(\reverse{R}\,R)\,Y\,\reverse{R}
 \\
 \;&=\; (R\,X\,\reverse{R})\,(R\,Y\,\reverse{R})
 \;=\; \mathrm{Ad}_R(X)\,\mathrm{Ad}_R(Y).
\end{aligned}
\label{eq:adR-mult}
\end{equation}
On the grade-one generators of $\Cl(1,3)$ it is the covering identity
\begin{equation}
 R\,\gamma_\mu\,\reverse{R}
 \;=\; \Lambda(R)^\nu{}_\mu\,\gamma_\nu
 \;\in\;\grade{\Cl(1,3)}{1},
 \label{eq:adR-grade1}
\end{equation}
where $\Lambda(R)\in\SOgroup^+(1,3)$ is the Lorentz transformation
assigned to $R$ by the two-to-one homomorphism
$\Spin^+(1,3)\to\SOgroup^+(1,3)$~\cite{Doran:2003}. Together,
Eqs.~\eqref{eq:adR-mult} and~\eqref{eq:adR-grade1} give, for an
$r$-blade $P_r=p_1\wedge\cdots\wedge p_r$ of grade-one elements
$p_a=p_a^\mu\gamma_\mu$,
\begin{equation}
 \mathrm{Ad}_R\bigl(p_1\wedge\cdots\wedge p_r\bigr)
 \;=\; (\Lambda p_1)\wedge\cdots\wedge(\Lambda p_r)
 \;\in\;\grade{\Cl(1,3)}{r}:
 \label{eq:adR-blade}
\end{equation}
$\mathrm{Ad}_R$ acts on every grade-one factor as $\Lambda$. It also
commutes with the antisymmetrisation
$p_1\wedge\cdots\wedge p_r=
\tfrac{1}{r!}\sum_\sigma\mathrm{sgn}(\sigma)\,
p_{\sigma(1)}\cdots p_{\sigma(r)}$,
so the right-hand side of Eq.~\eqref{eq:adR-blade} is again a
grade-$r$ blade. Each grade subspace $\grade{\Cl(1,3)}{k}$ of
$\Cl(1,3)$ is the linear span of its $k$-blades, so by
$\mathbb R$-linearity $\mathrm{Ad}_R$ preserves every grade.\par\endgroup

\section{Tables}
\label{app:tables}

\subsection{Full per-grade dictionary of classical observables}
\label{app:tables-c1}

Table~\ref{tab:grade-counts} counts the kinematic slots per grade of
$\evMV$ as a function of the final-state multiplicity $N$, and
Table~\ref{tab:full-dict} gives the dictionary. Here $N$ is the
number of final-state fermions after the decays $t\!\to\!bW$,
$W\!\to\!f\bar f'$, with the neutrinos replaced by a single $\MET$
pseudo-object.

\begin{table}[h!]
\centering
\footnotesize
\renewcommand{\arraystretch}{1.15}
\begin{tabular}{@{}l r r r r r r@{}}
\hline\hline
\textbf{$N$} & \textbf{grade 0} & \textbf{grade 1} &
\textbf{grade 2} & \textbf{grade 3} & \textbf{grade 4} &
\textbf{Total} \\
 & ($\binom{N}{2}\!+\!N$) & ($N$) &
($\binom{N}{2}$) & ($\binom{N}{3}$) & ($\binom{N}{4}$) & slots \\
\textit{components / slot $=\!\binom{4}{k}$:}
 & $1$ & $4$ & $6$ & $4$ & $1$ & \\
\hline
$4$ (minimal example) & $10$ & $4$ & $6$ & $4$ & $1$ & $25$ \\
$6$ (e.g.\ $t\bar t$ semileptonic) & $21$ & $6$ & $15$ & $20$ & $15$ & $77$ \\
$12$ (e.g.\ $t\bar t t\bar t$ all-hadronic) & $78$ & $12$ & $66$ & $220$ & $495$ & $871$ \\
\hline\hline
\end{tabular}
\renewcommand{\arraystretch}{1.0}
\caption{Number of kinematic \emph{slots} per grade
of $\evMV$ as a function of the final-state multiplicity $N$
(fermion-level, $\nu\!\to\!\MET$). A \emph{slot} is
\emph{one $\Cl(1,3)$ multivector} of the given grade attached
to a subset of the $N$ tokens. Its $\Vflav$ label,
$|f_{i_1}\rangle\!\otimes\!\cdots\!\otimes\!|f_{i_k}\rangle$, is
inherited from those tokens and is not counted separately. The counts
are combinatorial (second header row): on-shell masses plus pairwise
inner products at grade zero, four-momenta at grade one and wedge
products at grades two to four. Each slot carries $\binom{4}{k}$ real
components (third header row). The raw component count at grade
$k$ is therefore (slots)\,$\times\binom{4}{k}$,
e.g.\ for $N\!=\!6$: grade $1\!=\!6\!\times\!4\!=\!24$ numbers,
grade $2\!=\!15\!\times\!6\!=\!90$, grade $3\!=\!20\!\times\!4\!=\!80$,
total $230$.}
\label{tab:grade-counts}
\end{table}

The counts are formal. Each grade of $\Cl(1,3)$ has finite
dimension, $\dim\grade{\Cl(1,3)}{k}=
\binom{4}{k}\!\in\!\{1,4,6,4,1\}$ for $k\!=\!0,\ldots,4$, so for
$N\!>\!4$ the $\binom{N}{k}$ wedge products are linearly dependent,
since all lie in a $\binom{4}{k}$-dimensional subspace. Their span
saturates at $N\!=\!4$ (e.g.\ the single pseudoscalar direction at
grade 4), and each extra entry adds the value of its own $k$-tuple.
At grade four these are the $\lambda_{ijkl}$ pseudoscalars, whose mutual
relations underlie
Theorem~\ref{lemma:cm-collapse}. A network also need not consume the
full combinatorial set, which grows as $O(N^4)$, while the resonance
topology of $tWb$ fits in $O(N)$ tokens. Attention layers can select a
physics-motivated subset (e.g.\ $b$-jets paired with their
candidate $W$-decay products) or a top-$K$ truncation by norm. The
network of Sec.~\ref{sec:tWb} instead receives two pairing tokens
built by hand.

\renewcommand{\arraystretch}{1.12}
\footnotesize
\setlength{\tabcolsep}{3pt}
\newcolumntype{R}[1]{>{\raggedright\arraybackslash}p{#1}}
\begin{longtable}{@{}r R{4.3cm} R{3.1cm} R{2.3cm} R{3.8cm}@{}}
\caption{Full per-grade dictionary of $32$ classical observables of
hadron-collider events, by the grade of the multivector channel of
$\evMV$ and the operation that extracts them. Citations mark entries
with a specific origin. Column ``$\grade{\evMV}{0}$?'' tells whether
the observable reduces to the grade-zero sector by
Theorem~\ref{lemma:cm-collapse}: ``Y'' invariant reduction,
``$\Hlhc$'' reduction after contraction with the beam axis, ``N'' new
content, ``n.a.'' not a blade in $\Cl(1,3)$.}
\label{tab:full-dict}\\
\hline\hline
\textbf{\#} & \textbf{Observable} &
\textbf{Grade $\times$ operation} & \textbf{$\grade{\evMV}{0}$?} &
\textbf{Physical meaning} \\
\hline
\endfirsthead
\multicolumn{5}{c}{\textit{Table~\ref{tab:full-dict} continued from previous page.}}\\
\hline\hline
\textbf{\#} & \textbf{Observable} &
\textbf{Grade $\times$ operation} & \textbf{$\grade{\evMV}{0}$?} &
\textbf{Physical meaning} \\
\hline
\endhead
\hline
\multicolumn{5}{r}{\textit{Continued on next page.}}\\
\endfoot
\hline\hline
\endlastfoot
1 & $\hat s = (\textstyle\sum_i p_i)^2$
 & $0 \times \grade{(\sum p)^2}{0}$
 & Y
 & Partonic CM energy squared. \\
2 & $\hat t = (p_a-p_c)^2$
 & $0 \times \grade{\cdot}{0}$
 & Y
 & Mandelstam $t$-channel exchange. \\
3 & $\hat u = (p_a-p_d)^2$
 & $0 \times \grade{\cdot}{0}$
 & Y
 & Mandelstam $u$-channel exchange. \\
4 & $m_S^2 = (\textstyle\sum_{i\in S} p_i)^2$
 & $0 \times \grade{\cdot}{0}$
 & Y
 & Invariant mass of any subset $S$; $W,Z,H$, top and $t\bar t$ peaks. \\
5 & $m_i^2 = p_i \cdot p_i$
 & $0 \times \grade{p_i p_i}{0}$
 & Y
 & On-shell mass of object $i$. \\
6 & $\cos\theta^*$ spin projections (helicity and optimal bases)~\cite{Mahlon:1995zn, Mahlon:1996pn, Mahlon:1999gz}
 & $0$ via rotor sandwich + dot
 & Y
 & Decay-angle spin polarimeters in the top rest frame. \\
7 & $\cos\theta_{CS}$ (Collins--Soper)
 & $0$ via rotor sandwich + dot
 & Y
 & Lepton-pair angle; $A_{FB}$ sensitivity. \\
8 & Dalitz-plot variables $(m_{ij}^2,\, m_{jk}^2)$~\cite{Dalitz:1953cp}
 & $0 \times \grade{(p_i+p_j)^2}{0}$ on $\det G_3 = 0$
 & Y
 & 3-body decay phase space; Cayley--Menger boundary. \\
9 & $p_T^f,\, \eta^f,\, \phi^f,\, y^f,\, M_T^f$\textsuperscript{a}
 & $1$ tensored with $\hat n_z$
 & $\Hlhc$
 & Lab-frame transverse momentum, pseudorapidity, azimuth,
 rapidity and transverse mass of object $f$;
 $\SOgroup(2)_\phi\!\times\!\SOgroup^+(1,1)_z$ basis. \\
10 & $\Delta R_{ij} = \sqrt{\Delta\eta^2 + \Delta\phi^2}$
 & $1$ via beam-tensored differences
 & $\Hlhc$
 & Catchment-cone distance. \\
11 & $\Delta\phi_{ij}$
 & $1$ via $B^{12}$ projection
 & $\Hlhc$
 & Azimuthal opening angle. \\
12 & $\Delta y_{ij}$
 & $1$ via $B^{0z}$ projection
 & $\Hlhc$
 & Rapidity difference; longitudinal-boost invariant. \\
13 & $H_T = \sum_{\text{jets}} p_T$
 & $1$, sum of transverse norms
 & $\Hlhc$
 & Hadronic activity scalar sum. \\
14 & $S_T = H_T + \sum p_T^\ell + \MET$
 & $1$, sum of transverse norms
 & $\Hlhc$
 & Total transverse energy. \\
15 & $\MET$ pseudo-object
 & $1$ with masked $p_z$ component
 & $\Hlhc$-only
 & Sum of invisible transverse momenta. \\
16 & $B_W = p_\ell \wedge p_\nu$~\cite{Brehmer:2024yqw, Spinner:2024hjm}
 & $2$ via outer product
 & raw: N; $|B_W|^2$: Y
 & Oriented decay-plane bivector of $W$. \\
17 & $B_t = (p_\ell+p_\nu) \wedge p_b$~\cite{Brehmer:2024yqw}
 & $2$ via outer product
 & raw: N; $|B_t|^2$: Y
 & Oriented decay-plane bivector of $t$. \\
18 & $C_{ij}$ (Bernreuther basis $k,r,n$)~\cite{Bernreuther:2013aga, Bernreuther:2015yna}
 & $0$ via inner products with the spin axes, event average
 & $\Hlhc$
 & $t\bar t$ spin correlation. \\
19 & $S^{ij}$ (sphericity tensor)
 & symmetric rank-2; not in $\bigwedge^2$
 & n.a.
 & Eigenspectrum of momentum tensor; null-row. \\
20 & Spinor-helicity $\langle ij\rangle,\, [ij]$~\cite{Dixon:1996wi, Mangano:1990by, Maitre:2007jq}
 & $\Spin^+(1,3)\!\simeq\!\SLgroup(2,\mathbb C)$ inside the even
 subalgebra $\Cl^+(1,3)\!\simeq\!M_2(\mathbb C)$;
 not a blade in $\bigwedge^\bullet(\mathbb R^{1,3})$
 & n.a.
 & Massless QCD amplitude pair products; spinor lift of
 $\Spin^+(1,3)$. \\
21 & $S,\, A,\, P$ (sphericity, aplanarity, planarity)
 & $0$ via eigenvalues of $S^{ij}$
 & $\SOgroup(3)_{\rm lab}$ only (not boost-inv.)
 & Event-shape scalars. \\
22 & Thrust $T$
 & $0$ via $\max_{\hat n}\sum |p_i\!\cdot\!\hat n|$
 & $\SOgroup(3)_{\rm lab}$ only (not boost-inv.)
 & Maximal alignment scalar. \\
23 & Fox--Wolfram moments $H_\ell$
 & $0$ via $\sum |p_i||p_j| P_\ell(\cos\theta_{ij})$
 & $\SOgroup(3)_{\rm lab}$ only (not boost-inv.)
 & Spherical-harmonic moments. \\
24 & $n$-subjettiness $\tau_N$~\cite{Thaler:2010tr}
 & $0$ via jet-axis projections
 & $\SOgroup(3)_{\rm lab}$ only (not boost-inv.)
 & Jet-substructure axis-counting ($e^+e^-$ definition). \\
25 & Soft-drop mass $m_{\text{SD}}$~\cite{Larkoski:2014wba}
 & $0$ via $(y,\phi)$ jet-grooming
 & $\Hlhc$-only
 & Groomed jet mass. \\
26 & $\vec p_a \cdot (\vec p_b \times \vec p_c)$~\cite{Atwood:2000tu}
 & $3$ via $T^{123}$ component of $p_a\wedge p_b\wedge p_c$
 & sign: N; $|T|^2$: Y
 & P-odd triple product (motion-reversal-T-odd in the sense of
 Ref.~\cite{Atwood:2000tu}; covariant-T-even,
 cf.\ Table~\ref{tab:symm-spacetime-disc}); CP-violation probe. \\
27 & $\det G_3 = \det[p_a\!\cdot\!p_b]_{a,b\in\{i,j,k\}}$
 & $0$ via $\grade{T \reverse T}{0}$
 & Y
 & Squared 3-volume in Minkowski. \\
28 & $\det G_4 = \det[p_a\!\cdot\!p_b]_{a,b\in\{i,j,k,l\}}$
 & $0$ via $\grade{P \reverse P}{0}$
 & Y
 & Squared 4-volume magnitude. \\
29 & $\mathrm{sgn}\grade{p_1 p_2 p_3 p_4}{4}$ (this work)~\cite{Atwood:2000tu, Bernreuther:2015yna}
 & $4$ via $\mathrm{sgn}(\lambda)$, $\lambda I = p_1\wedge\ldots\wedge p_4$
 & N (new)
 & Sign of oriented 4-volume; CP-odd. \\
30 & Bernreuther $T_5,\, B_3,\, \mathcal A_{CP}$~\cite{Bernreuther:2015yna}
 & $4$ via $\mathrm{sgn}$ grade-4 / $3$ raw
 & mostly via $\mathrm{sgn}\grade{\cdot}{4}$: N
 & CP-odd / T-odd correlations in $t\bar t$ decays. \\
31 & Forward--backward asymmetry $A_{FB}$
 & $1$ via $\mathrm{sgn}(p_z\!\cdot\!\hat n_{\text{dir}})$
 & $\Hlhc$
 & $\hat z$-projection sign asymmetry. \\
32 & Charge-asymmetry $Q_\ell\,\eta_{\text{fj}}$~\cite{AguilarSaavedra:2008zc}
 & $\Vflav \otimes \grade{\Cl(1,3)}{1}$:
 $Q_\ell$ from $\Vflav$, $\eta_{\text{fj}}$ from grade-1
 & $\Hlhc$
 & Lepton-charge $\times$ forward-jet pseudorapidity;
 single-top tag. \\ \\
\end{longtable}
\normalsize
\renewcommand{\arraystretch}{1.0}

\noindent\textsuperscript{a} For neural-network inputs we recommend
the embedding $(\cos\phi^f,\sin\phi^f)$ in place of the raw angle.

\paragraph{Auxiliary and null rows.}
The grade of each row of Table~\ref{tab:full-dict} is given in its
third column. Rows 21--25
lie outside $\overline{\evMV}$ and enter as auxiliary grade-zero
inputs (event-shape paragraph of Sec.~\ref{sec:grades}). Rows 19 and 20,
the sphericity tensor $S^{ij}$ and the spinor-helicity pair
$\langle ij\rangle,\,[ij]$, are not blades in $\Cl(1,3)$.
They are listed because a reader may look for them. The
spinor-helicity brackets are bilinears in the even subalgebra
$\Cl^+(1,3)\simeq M_2(\mathbb C)$.

\subsection{Symmetry inventory}
\label{app:tables-c2}

\subsubsection{Continuous spacetime symmetries}
\label{app:tables-c2-cont}

\begin{table}[h!]
\centering
\footnotesize
\caption{Continuous spacetime symmetries.}
\label{tab:symm-spacetime-cont}
\renewcommand{\arraystretch}{1.15}
\begin{tabularx}{\linewidth}{@{}l p{2.6cm} p{2.4cm} X p{1.7cm} p{2.4cm}@{}}
\hline\hline
\textbf{\#} & \textbf{Name} & \textbf{Group} &
\textbf{Action on $\Cl(1,3)$} &
\textbf{On $\Vflav$} & \textbf{Enforcement} \\
\hline
A.1 & Lorentz proper orthochr. & $\Spin^+(1,3)$\newline$\simeq\SLgroup(2,\mathbb C)$
 & $X \mapsto R\, X\, \reverse R$ (rotor sandwich)
 & trivial
 & strict on parton; approx.\ on reco \\
A.2 & Translation $T(4)$ & $\mathbb R^4$ abelian
 & trivial on 4-momenta (Noether-conserved);
 $x_i\!\mapsto\!x_i+a$ on positions in PGA $\Cl(1,3,1)$
 (out of scope of $\evMV$)
 & trivial
 & out (mention only) \\
A.3 & Conformal & $\mathrm{SO}(2,4)$, 15-D
 & $\Cl(2,4)$ extension
 & trivial
 & out (mention only) \\
A.4 & Dilatation & $\mathrm{SO}(1,1)_D$
 & grade-preserving scale
 & trivial
 & soft (layer-norm) \\
A.5 & Beam-axis residual $\Hlhc$ & $\SOgroup^+(1,1)_z$\newline$\times\SOgroup(2)_\phi$
 & two rotor sandwiches
 & trivial
 & strict (incl.\ $\MET$) \\
A.6 & $2{+}2$ algebraic split & $\Cl(1,1)_z$\newline$\hat\otimes\,\Cl(0,2)_\perp$
 & basis tensor decomposition
 & trivial
 & strict (basis choice) \\
\hline\hline
\end{tabularx}
\renewcommand{\arraystretch}{1.0}
\end{table}

\subsubsection{Discrete spacetime symmetries}
\label{app:tables-c2-disc}

\paragraph{Two conventions for $T$.}
Two definitions of the Lorentz-improper time-reversal element are in
common use, and the paper uses both. The \emph{Wigner motion-reversal}
convention $\Lambda_T^{\mathrm{W}}=\mathrm{diag}(+1,-1,-1,-1)$ of the
``T-odd triple products'' of Refs.~\cite{Atwood:2000tu, Bernreuther:2013aga,
Bernreuther:2015yna} maps $(E,\vec p)\!\to\!(E,-\vec p)$. On four-vectors it
is the parity matrix, and it
differs from $P$ only by the antiunitarity of $T$. The \emph{covariant
T} convention $\Lambda_T^{\mathrm{cov}}=\mathrm{diag}(-1,+1,+1,+1)$ of
Sec.~\ref{sec:symm} and of the Grade~4 paragraph of
Sec.~\ref{sec:grades} (Eq.~\eqref{eq:pseudoscalar}) maps
$(E,\vec p)\!\to\!(-E,+\vec p)$. As $4\!\times\!4$ matrices the two
differ by an overall sign,
$\Lambda_T^{\mathrm{cov}}=-\Lambda_T^{\mathrm{W}}$, and both
have $\det=-1$, so they agree on every Lorentz-improper one-bit
channel. The pseudoscalar
$\lambda=\epsilon_{\mu\nu\rho\sigma}\,p_i^\mu p_j^\nu p_k^\rho p_l^\sigma$
is T-odd in either convention, and the CP-odd channel
$\mathrm{sgn}\,\lambda$ (row G.1 of Table~\ref{tab:symm-approx-ml-cp})
does not depend on the choice. The conventions disagree on grade-three
spatial observables. The triple product
$\vec p_a\!\cdot\!(\vec p_b\!\times\!\vec p_c)$ is
motion-reversal-T-odd, because its three momenta flip sign under
$\Lambda_T^{\mathrm{W}}$, and covariant-T-even, because
$\Lambda_T^{\mathrm{cov}}$ leaves them unchanged. It is P-odd in both.
Table~\ref{tab:symm-spacetime-disc} realises $T$ in the covariant
convention, and the triple-product row of Table~\ref{tab:full-dict} keeps the
motion-reversal ``T-odd'' label of Refs.~\cite{Atwood:2000tu,
Bernreuther:2015yna} with a note. We use the covariant choice because
$\Lambda_T^{\mathrm{cov}}$ acts
non-trivially on $\gamma_0$ and trivially on $\gamma_i$. The
T-channel then decouples from the parity channel B.1 at the level of
the Clifford action.

\begin{table}[h!]
\centering
\footnotesize
\caption{Discrete spacetime symmetries, with T in the covariant
convention (see the two conventions for $T$ above). The
time-reversal element is the spatial pseudoscalar
$\gamma_T \equiv \gamma_1\gamma_2\gamma_3$ of $\Cl(1,3)$, real, with
$\gamma_T^2=+1$ in the $(+,-,-,-)$ signature. The antiunitary action
of row B.2 includes complex conjugation (``c.c.''). The symbol $C_V$
is charge conjugation on $\Vflav$, which exchanges particle and
antiparticle labels.}
\label{tab:symm-spacetime-disc}
\renewcommand{\arraystretch}{1.15}
\begin{tabularx}{\linewidth}{@{}l p{2.4cm} p{2.0cm} X p{2.0cm} p{2.5cm}@{}}
\hline\hline
\textbf{\#} & \textbf{Name} & \textbf{Group} &
\textbf{Action on $\Cl(1,3)$} &
\textbf{On $\Vflav$} & \textbf{Enforcement} \\
\hline
B.1 & Parity P & $\mathbb Z_2$
 & $X\mapsto\gamma_0 X\gamma_0$
 & trivial
 & not enforced (preserves CP-odd channel) \\
B.2 & Time reversal T & $\mathbb Z_2$ antiunitary
 & $X\mapsto\gamma_T X\gamma_T^{-1}$ + c.c.\ ($\gamma_T^2=+1$,
 so $\gamma_T^{-1}=\gamma_T$; covariant-T convention,
 $\Lambda_T^{\mathrm{cov}}\!=\!\mathrm{diag}(-1,+1,+1,+1)$)
 & trivial
 & not enforced (broken via CPT$\oplus$CP) \\
B.3 & Charge conjugation C & $\mathbb Z_2$
 & trivial
 & $f\to C_V f$, $f^+\!\leftrightarrow\!f^-$
 & not enforced (charge-bit embedding) \\
B.4 & CP & $\mathbb Z_2$
 & $X\mapsto\gamma_0 X\gamma_0$
 & $f\to C_V f$
 & not enforced (CP-asymmetry head reads the channel,
 Sec.~\ref{sec:outlook}) \\
B.5 & CPT & identity (group)
 & $X\!\mapsto\!\hat X$ + c.c.\ (grade-involution
 $(-1)^r$ on grade-$r$ from $P\!\cdot\!T$; trivial on
 physical bilinears by the CPT theorem)
 & $f\!\to\!C_V f$ (inherited from C)
 & automatic (CPT theorem; not implemented) \\
\hline\hline
\end{tabularx}
\renewcommand{\arraystretch}{1.0}
\end{table}

\subsubsection{Internal and permutational symmetries}
\label{app:tables-c2-int}

\begin{table}[h!]
\centering
\footnotesize
\caption{Internal and permutational symmetries.}
\label{tab:symm-internal}
\renewcommand{\arraystretch}{1.15}
\begin{tabularx}{\linewidth}{@{}l p{2.4cm} p{2.4cm} l X p{2.0cm}@{}}
\hline\hline
\textbf{\#} & \textbf{Name} & \textbf{Group} &
\textbf{$\Cl(1,3)$} & \textbf{Action on $\Vflav$} &
\textbf{Enforcement} \\
\hline
C.1 & $\mathrm{U}(1)_{\text{em}}$ & $\mathrm{U}(1)$
 & trivial
 & charge-bit on $\Vflav$
 & embedding (active) \\
C.2 & EW $\mathrm{SU}(2)_L\!\times\!\mathrm{U}(1)_Y$ & 4-D
 & trivial
 & not in feature space
 & out (handled by MC) \\
C.3 & Colour $\mathrm{SU}(3)_C$ & 8-D
 & trivial
 & not in feature space
 & out (confined; MC) \\
C.4 & CKM unitarity & $\mathrm{U}(3)/\text{phases}$ (4 params)
 & trivial
 & not in feature space
 & out (handled by MC) \\
C.5 & PMNS unitarity & $\mathrm{U}(3)/\text{phases}$ (4 or 6 params)
 & trivial
 & not in feature space
 & out (handled by MC) \\
D.1 & $S_n$ within flavour type & $\prod_a S_{n_a}$
 & trivial
 & tensor-symmetric on $\Vflav$
 & strict (set-attention) \\
D.2 & Bose/Fermi statistics & $\mathbb Z_2$ amplitude
 & trivial
 & implicit in MC
 & covered by D.1 \\
D.3 & LFU $S_3$ & $S_3$ on $e/\mu/\tau$
 & trivial
 & weight-tying option
 & soft (document break by masses) \\
D.4 & Generation $S_3$ (quarks) & $S_3$
 & trivial
 & strongly broken
 & not enforced (Yukawa) \\
D.5 & Isospin $\mathrm{SU}(2)_I$ & 3-D
 & trivial
 & confined hadron multiplets
 & out (broken by EW + masses) \\
\hline\hline
\end{tabularx}
\renewcommand{\arraystretch}{1.0}
\end{table}

\subsubsection{Approximate, ML-architectural and CP-channel symmetries}
\label{app:tables-c2-approx}

\begin{table}[h!]
\centering
\footnotesize
\caption{Approximate, ML-architectural and CP-channel symmetries.}
\label{tab:symm-approx-ml-cp}
\renewcommand{\arraystretch}{1.15}
\begin{tabularx}{\linewidth}{@{}l p{2.6cm} p{2.4cm} l X p{2.0cm}@{}}
\hline\hline
\textbf{\#} & \textbf{Name} & \textbf{Group} &
\textbf{$\Cl(1,3)$} & \textbf{Action on $\Vflav$} &
\textbf{Enforcement} \\
\hline
E.1 & Chiral $\mathrm{SU}(N)_L\!\times\!\mathrm{SU}(N)_R$ & $2(N^2{-}1)$
 & trivial
 & trivial
 & out (confined; S$\chi$SB) \\
E.2 & Heavy-quark symmetry & $\mathrm{SU}(2)_s\!\times\!\mathrm{U}(N_h)_f$
 ($N_h$ heavy flavours, e.g.\ $b,c$)
 & trivial
 & trivial
 & out (top decays first) \\
E.3 & Custodial $\mathrm{SU}(2)_V$ & 3-D
 & trivial
 & trivial
 & out (Higgs-sector specific) \\
F.1 & Architectural set-perm. & $S_n$
 & trivial
 & tensor-symmetric (set-attention)
 & strict (by construction) \\
F.2 & Token positional shift & $\mathbb Z_n$
 & trivial
 & none (no positional encoding)
 & conditional (absent-by-design) \\
F.3 & Feature-scale & $\mathbb R_+$
 & trivial
 & layer-norm
 & soft (architectural) \\
G.1 & Pseudoscalar sign $\mathbb Z_2$ & $\mathbb Z_2$
 & $\mathrm{sgn}\grade{p_1 p_2 p_3 p_4}{4}\to-\mathrm{sgn}$
 & trivial
 & conditional (CP-asymmetry head input,
 Sec.~\ref{sec:outlook}) \\
G.2 & $\mathcal A_{CP}$ in decay rates & $\mathbb Z_2$
 & trivial
 & charge-bit anti-symmetrisation
 & embedding (implicit) \\
G.3 & EDMs & $\mathbb Z_2$
 & trivial
 & trivial
 & out (low-energy, not LHC) \\
\hline\hline
\end{tabularx}
\renewcommand{\arraystretch}{1.0}
\end{table}


\section{Token layouts of the two input branches}
\label{app:tokens}

A token (Sec.~\ref{sec:intro}) is a stack of \emph{channels}, each a
full $16$-component multivector of $\Cl(1,3)$. The
$1{+}4{+}6{+}4{+}1$ cells of those $16$ components follow the basis
of the algebra: one scalar cell, the four $\gamma_\mu$, the six
coordinate planes $\gamma_{\mu\nu}$, the four $\gamma_{\mu\nu\rho}$ and
the pseudoscalar $\gamma_{0123}$. The cells store the five grades, and
which of them are populated is the input choice under test.
A single four-momentum fills the four grade-one cells, and a nonzero
higher-grade object needs several particles, since $p_i\wedge p_i=0$.
The equivariant linear layers mix channels with one weight per grade,
five per channel pair. This accounts for the whole
$2\,080$-parameter difference between the two branches of
Table~\ref{tab:tWb-auc} ($13$ extra input channels $\times\,32$
outputs $\times\,5$).

The layouts store a slot $p_i\otimes|f_i\rangle$ of
Eq.~\eqref{eq:evMV}, a \emph{simple} tensor, by its two factors: the
momentum in the grade-one cells of one channel and the label as a
one-hot across grade-zero cells. For a basis vector $|f_i\rangle$ this
is in bijection with the dense $4\times\dim\Vflav$ array of
coefficients. The dense form would be needed only for superpositions
of type labels, which no reconstruction produces.

\paragraph{Four-momentum branch (grade $0\oplus1$, $6$--$8$ tokens,
$9$ channels).} A particle token fills $12$ of its $9\times16$ cells:
\begin{center}
\begin{tabular}{llll}
\toprule
channel & grade cells & content & meaning \\
\midrule
$0$ & $1$ (four cells) & $(E,\,p_x,\,p_y,\,p_z)_i$ & measured momentum \\
$1$--$6$ & $0$ (one cell each) & type one-hot & object identity \\
$7$ & $0$ & electric charge & \\
$8$ & $0$ & $b$-tag & \\
\bottomrule
\end{tabular}
\end{center}
All grade-$2,3,4$ cells enter as zeros. The network fills them by
attention across tokens and geometric products within a token
(Sec.~\ref{sec:ablation}), and after $L$ blocks the products reach
depth $\sim2^L$, enough for the network to build the products
$\overline{\evMV}$ of Sec.~\ref{sec:evmv-def}.
The reference rows of Table~\ref{tab:tWb-auc} append tokens of the
same shape, each carrying one fixed element of the algebra in its
covariant channel. These are two vector tokens for $\gamma_0$ and
$\gamma_3$, or one token with the grade-two cell $\gamma_{03}$ set.
Constants need no
trainable embedding, so these rows have the same parameter count.

\paragraph{Pairing branch ($8$ tokens, $22$ channels).} The six
particle tokens keep their content in a wider layout. Their channel
$1$ and the ten pairing-scalar channels are zero, and the type one-hot
spans eight slots, the six particle species plus one sentinel slot per
pairing token. The two event-level pairing tokens, one per candidate
$b$-assignment $a\in\{1,2\}$, carry the topology candidates in grades
zero to three:
\begin{center}
\begin{tabular}{llll}
\toprule
channel & grade cells & content & meaning \\
\midrule
$0$ & $1$ (four) & $\star T_{\rm had}^{(a)}$ & dual of hadronic-top volume \\
$0$ & $2$ (six) & $T_{\rm had}^{(a)}\cap T_{\rm lep}^{(a)}$ & meet of the two decay volumes \\
$0$ & $3$ (four) & $T_{\rm had}^{(a)}=p_{j_1}\wedge p_{j_2}\wedge p_{b_a}$ & hadronic-top candidate volume \\
$1$ & $1\oplus3$ & $\star T_{\rm lep}^{(a)},\;T_{\rm lep}^{(a)}$ & leptonic side \\
$12$--$21$ & $0$ & pulls, $\sigma^{\pm}_{Wb}$, $m_W$, coplanarity, $\mathrm{sgn}\,\lambda$ & topology invariants \\
\bottomrule
\end{tabular}
\end{center}
In both branches the event is one array of slots $\times$ channels
$\times$ $16$ grade cells, and only the populated cells differ.

\bibliographystyle{cpc}
\bibliography{references}

\end{document}